\documentclass[12pt]{article}
 \usepackage{amsmath}
 \usepackage{graphicx}
 \usepackage{amssymb}
 \usepackage{amsfonts}

 \def\p{\partial}
 \def\ubar{\bar{u}}
 \def\ub{\mathbf{u}}
 \def\ubarb{\bar{\mathbf{u}}}
 \def\xb{\mathbf{x}}
 \def\phib{\boldsymbol{\phi}}
 
 \def\mb{\mathbf{m}}
 \def\kb{\mathbf{k}}
 \newtheorem{theorem}{Theorem}[section]

 \newenvironment{proof}[1][Proof]{\begin{trivlist}\item[\hskip \labelsep
{\bfseries #1}]}{\end{trivlist}}

 \newcommand{\centre}[2]{\multispan{#1}{\hfill #2\hfill}}

\newcommand{\eref}[1]{ (\ref{#1}) }
\newcommand{\dotted}{\protect\mbox{${\mathinner{\cdotp\cdotp\cdotp\cdotp\cdotp\cdotp}}$}}
\newcommand{\dashed}{\protect\mbox{-\; -\; -\; -}}
\newcommand{\chain}{\protect\mbox{--- $\cdot$ ---}}
\newcommand{\dashddot}{\protect\mbox{--- $\cdot$ $\cdot$ ---}}
\newcommand{\full}{\protect\mbox{------}}

\title{A Regularization of Burgers Equation using a
Filtered Convective Velocity}

\author{
 Greg Norgard %
    \thanks{Graduate Student,Department of Applied Mathematics,
   526 UCB.} 
 and Kamran Mohseni%
   \thanks{Associate Professor, Department of Aerospace Engineering Sciences
 and RECUV,Room ECAE 175, and AIAA Member.}\\
  {\normalsize\itshape
  University of Colorado, Boulder, Colorado, 80309, US}
 }

\begin{document}
\maketitle

\vspace*{10pt}
\indent
\textit{\footnotesize Manuscript prepared for submission to the Journal of Physics A, special edition D$^2$HFest based on a related presentation at Darryl D. Holm 60th Birthday Celebration, Lausanne Switzerland, July 22-28, 2007.}

\begin{abstract}
This paper examines the properties of a regularization of Burgers equation in one and multiple dimensions using a filtered convective velocity, which we have dubbed as convectively filtered Burgers (CFB) equation.  A physical motivation behind the filtering technique is presented.  An existence and uniqueness theorem for multiple dimensions and a general class of filters is proven.  Multiple invariants of motion are found for the CFB equation and are compared with those found in viscous and inviscid Burgers equation. Traveling wave solutions are found for a general class of filters and are shown to converge to weak solutions of inviscid Burgers equation with the correct wave speed.  Accurate numerical simulations are conducted in 1D and 2D cases where the shock behavior, shock thickness, and kinetic energy decay are examined. Energy spectrum are also examined and are shown to be related to the smoothness of the solutions.

\end{abstract}

\maketitle

\section{Introduction}
The Euler and Navier-Stokes equations are well known as the fundamental laws governing fluid dynamics, yet even after 200 years they continue to present theoretical and computational challenges.  The nonlinear terms inherent in the equations give rise to small scale structures, in the form of turbulence and shocks, which have proven to be the bane of computational simulations.  With the proper modeling of these small scales, we hope it is possible to address both turbulence and shocks with one encompassing technique.  Currently, the Lagrangian averaging approach is making strides in handling turbulent flows \cite{Holm:04a,Foias:01a,FoiasC:02a,Marsden:01h,Mohseni:03a,Mohseni:04c,Mohseni:05e,Mohseni:06l,Mohseni:05d,Marsden:98b,Chen:99b}. That work motivated the work presented in this paper. The Lagrangian averaging approach results in a filtered convective velocity in the nonlinear term.  This paper also uses a filtered convective velocity in the nonlinear term of Burgers equation, with the intention of discovering if this technique is a reasonable means of capturing shock formation.
%

Burgers equation has been a useful testing grounds for fluid dynamics for many
years due to the fact that it shares the same nonlinear convective term as Euler
and Navier-Stokes equations.  Indeed, Burgers equation has been the focus of much work, both numerically and analytically \cite{BurgersJM:48a, LighthillMJ:56a, ColeJD:51a,WhithamGB:74a, LaxPD:73a, Kraichnan:93a, Holm:04b, TadmorE:04a, OberaiAA:06a}.  Originally Burgers equation was introduced and used as a simplistic model for one dimensional turbulence \cite{BurgersJM:48a, LighthillMJ:56a}.  The original equation,
\begin{equation}
\label{1Dviscousburgers}
u_t+uu_x=\nu u_{xx},
\end{equation}
which shall be referred to as viscous Burgers equation, includes a dissipative viscous term.  By removing the viscous term,
\begin{equation}
 \label{1Dinviscidburgers}
u_t+uu_x=0,
\end{equation}
inviscid Burgers equation is obtained. Much like the Euler equations, inviscid Burgers equation can be expressed as a conservation law. Both viscous and inviscid Burgers equation are easily extended into multiple dimensions giving
\begin{equation}
\label{ndviscousBurgers}
\ub_t+\ub \cdot \triangledown \ub=\nu \bigtriangleup \ub
\end{equation}
and
\begin{equation}
\label{ndinviscidBurgers}
\ub_t+\ub \cdot \triangledown \ub=0.
\end{equation}

It is well established that inviscid Burgers equation forms discontinuities in
finite time, determined by initial conditions \cite{WhithamGB:74a, LaxPD:73a}.  To deal with these discontinuities
weak solutions are introduced.  However, when weak solutions are
introduced, solutions are no longer necessarily unique \cite{LaxPD:73a,
OleinikOA:57a}.  In order to choose the physically relevant solution, an entropy
condition is applied, which one and only one weak solution satisfies. This
physically relevant solution is referred to as the entropy solution.  Lax,
Oleinik, and Kruzkov have examined the entropy condition for conservation laws and expressed it using
different techniques \cite{LaxPD:73a,
OleinikOA:57a, KruzkovSN:70a}.  Each of their entropy conditions can be used in different classes of conservation laws, but can all be applied to inviscid Burgers equation with equivalent results  \cite{LellisCD}.

With the dissipative term, viscous Burgers equation does not form
discontinuities \cite{ WhithamGB:74a,LaxPD:73a}.  Indeed it smooths
discontinuities found in initial conditions, and has a unique infinitely
differentiable solution for all time \cite{ColeJD:51a, HopfE:50a}.  Moreover,
the limit of the solutions as $\nu \to 0$ converge strongly to the entropy
solution of inviscid Burgers equation \cite{LaxPD:73a, OleinikOA:57a,
KruzkovSN:70a}.

Viscous Burgers equation is not the only regularization of inviscid Burgers
equation that converges to the entropy solution.  Many regularizations have been
proposed, typically with the addition of a dissipative term.  Among them are
regularizations with hyperviscosity, filtered viscosity, and a combination of
viscosity and dispersion
\cite{TadmorE:04a,LattanzioC:03a,TadmorE:01a,TadmorE:92a,KondoCI:02a}.

Another famous regularization is the KDV equation which uses a dispersive term,
\begin{equation}
\label{kdv}
u_t+u\,u_x=-\epsilon u_{xxx}.
\end{equation}
This regularizes inviscid Burgers equation in the sense that solutions are now
continuous, however, many oscillations form as $\epsilon \to 0$, requiring a
weak limit for convergence \cite{Lax:83a,Kawahara:70b}.
This limit is not the entropy solution, nor even a weak solution of inviscid Burgers equation \cite{GurevichAV:74a}.

The incompressible isotropic LANS-$\alpha$ equations are given by \cite{Marsden:98b}
\begin{equation}
\label{lansa}
\frac{\p \ub}{\p t} + \ubarb\cdot \triangledown  \ub +u_j \triangledown \ubar_j = -\triangledown P +\frac{1}{Re}\bigtriangleup \ub
\end{equation}
where $\ub$ is defined as 
\begin{equation}
\ub=\ubarb-\alpha^2 \bigtriangleup \ubarb.
\end{equation}
As mentioned earlier this equation has had success capturing some of the behavior for a class of turbulent flows.  It can be seen  in Equation (\ref{lansa}) that LANS-$\alpha$ employs a filtered convective velocity.  What this paper proposes is that it is this filtered convective velocity in the nonlinear term that is successful in modeling small scale behavior.  If this is indeed the case, then such a term could also be used to model shocks alongside turbulence. This paper begins the examination of using such a nonlinear term in shock regularization.  As was stated, Burgers equation has been extensively researched and is well understood.  It forms shocks readily and has been chosen as a testing ground for using a filtered convective velocity for shock regularization.  This technique is not intended to be an analytically or numerically superior method of regularizing shocks, but more as a proof of concept with the intention of extending this method into areas involving both shocks and turbulence.

This paper considers what will be referred to as convectively filtered Burgers
equation (CFB), where the convective velocity in inviscid Burgers equation is
replaced with a filtered (averaged) velocity.  The filtering is done by
convoluting the velocity $u$ with the kernel $g$ (properties of which will be
discussed in section \ref{filters}) resulting in 
\begin{eqnarray}
\label{CFB1Da}
u_t+\ubar u_x=0\\
\label{CFB1Db}
\ubar=g\ast u.
\end{eqnarray}

If an inverse for the filtering exists, the equation, using only the filtered velocity, can be rewritten as
\begin{equation}
\label{smoothonly}
\ubar_t+\ubar\ubar_x=\overline{g^{-1}(\ubar\ubar_x)-\ubar g^{-1}(\ubar_x)},
\end{equation}
where $g^{-1}$ represents the inverse of the low pass filter. Similar to the Large Eddy Simulation (LES) or the Lagrangian Averaged Navier-Stokes-$\alpha$ (LANS-$\alpha$) equations where the effect of small scales (high wavenumbers) on the large scales are modeled by a subgrid scale stress tensor \cite{Hughes:01a,Hughes:01b}. To this end the right hand side of Equation \eref{smoothonly} can be considered a convective subgrid scale stress for regularizing a shock. 

This being said, it is the intention that the filtered velocity, $\ubar$, be considered the physically relevant quantity, that captures the proper macroscopic, or low wavemode, behavior.  In Sections \ref{numerics}-\ref{energynorms}, $\ubar$ from CFB will be compared to the velocity from viscous Burgers equation. The unfiltered velocity, $u$ is intended to capture both the low and high wavemode activity.  In Section \ref{invariantsofmotion} most of the invariants of motion found for CFB involve the unfiltered velocity.

CFB, like Burgers equations, is easily extended into multiple dimensions, with equation

\label{CFB}
\begin{eqnarray}
\label{CFBa}
\ub_t+\ubarb\cdot \triangledown \ub=0\\
\label{CFBb}
\ubarb=g\ast \ub.
\end{eqnarray}

 Special attention is paid to the Helmholtz filter, which gains notice for
several reasons.  It is a common filter, employed in the Leray-$\alpha$ model of
turbulence \cite{Holm:04a,Foias:01a,FoiasC:02a}, Lagrangian Averaged Navier-Stokes (LANS-$\alpha$) \cite{Marsden:01h,Mohseni:03a,Mohseni:04b,Mohseni:04c,Mohseni:05e} and in the Lagrangian averaged Euler-$\alpha$ \cite{Mohseni:05d,Marsden:98b}. It has been shown that one dimensional CFB is Hamiltonian when
the Helmholtz filter is used \cite{Holm:03a, BhatHS:06a}. Also, the inverse of the
Helmholtz filter is known for vanishing or periodic boundary conditions, so it
can be expressed through convolution or using differential notation.  In 1D, the following are equivalent

\begin{eqnarray}
\label{helmholtzfiltera}
 \ubar=\frac{1}{2}e^{-|\cdot|} \ast u\\
\label{helmholtzfilterb}
 u =\ubar -  \ubar_{xx}.
 \end{eqnarray}

When moving to higher dimensions the convolution kernel changes and the double derivative becomes a Laplacian.   Using the Helmholtz filter, Equation \eref{smoothonly} can be simplified to
\begin{equation}
\label{aburgersles}
\ubar_t+\ubar\ubar_x=-3\alpha^2 (I-\alpha^2\p x^2)^{-1} (\ubar_x \ubar_{xx}).
\end{equation}

The CFB equation is not the first look into the use of filtered convective velocities.  Jean Leray proposed using a filtered convective velocity in the Navier-Stokes equations as early as 1934 \cite{LerayJ:34a}.  This was done in an attempt at proving that regular solutions to the Navier-Stokes equations exist and to examine properties of those solutions.  This has led to several investigations into the Leray-$\alpha$ model of turbulence \cite{Holm:04a, IlyinAA:06a, Hanjalic:06a} primarily using the Helmholtz filter. In the projected models for extending our technique into the Euler equations, the filtering occurs in more than just the convective velocity.  At that point the research loses much of the similarity with Leray's work.  This extension into Euler equations will be the topic of following papers.

The CFB equation itself has been previously studied.  The family of equations
\begin{eqnarray*}
u_t+\ubar u_x+b u \ubar_x=0,\\
\ubar=g\ast u
\end{eqnarray*}
where $b$ was a free parameter, has been examined \cite{Holm:03a,Holm:03b}, where $g$ was almost exclusively the Helmholtz filter.  The $b=0$ case reduces to the 1D version of CFB.  Holm and Staley \cite{Holm:03a} established the invariant quantities, $\int u$ and $||u||_\infty$.  They examined peakon solutions and cliff solutions both analytically and numerically.  A Lagrangian structure was shown for the $b\neq 0$ case and a Hamiltonian structure proposed for $b=0$ by Degasperis et al. \cite{Holm:03b}.  Bhat and Fetecau \cite{BhatHS:06a} also examined one dimensional CFB with the Helmholtz filter.  They more fully established the Hamiltonian structure proposed by Degasperis et al.  They also proved existence and uniqueness of a solution, and proved that the solutions converge to a weak solution of inviscid Burgers equation.

This paper extends the results into the use of more general filters and looks into 1D and multiple dimensions.  It also examines many of the physically relevant characteristics of CFB.  Some of the results in this paper have been previously presented in conference papers \cite{Mohseni:06l,Mohseni:07e,Mohseni:07s}. The next section elaborates on the motivation for the averaging (filtering) technique.  Section \ref{filters} establishes the characteristics desired in the filter.   Existence and uniqueness theorems are given in section \ref{ExistenceTheorem} for CFB in multiple dimensions and a general class of filters.  Section \ref{invariantsofmotion} examines the various invariants of motion for CFB and compares them to those found in viscous and inviscid Burgers equation. In section \ref{travelingwavesolution}, a traveling wave solution for general filters is presented and analyzed. Sections \ref{numerics}-\ref{energynorms} deal with numerical simulations and results regarding shock behavior, spectral energy, and energy norms.  All is then followed with concluding remarks.

\section{Shock Formation}\label{shockformation}
We begin by examining the mechanics behind shock formation. The nonlinear convective term $\ub\cdot\triangledown \ub$ generates high wave
modes, by transferring energy into smaller scales as time progresses.  This term
can be found in the Euler and Navier-Stokes equations, where it is responsible
for tilting and stretching of vortical structures \cite{BradshawP:76a}.  It is also found in Burgers
equation.  If this cascade of energy into the smaller scales is left unchecked,
discontinuities will form.  To prevent this, viscosity is added.  Dependent on the
Reynolds number, there exists a length scale where viscosity begins to
dominate the energy cascade, transferring kinetic energy into thermal energy
through dissipation.  In the Navier-Stokes equation, the energy cascade has a
slope of -$\frac{5}{3}$ until this length scale, known as the Kolmogorov
scale, illustrated in figure \ref{kscale}a.  Viscous Burgers
equation has an energy cascade slope of -$2$ until viscosity begins to dominate;
seen in figure \ref{kscale}b \cite{Kraichnan:68a, GurbatovSN:97a}.  The Euler
equations and inviscid Burgers equation are the respective limits as viscosity
goes to zero.  By filtering the convective velocity, the energy cascade is
affected through the nonlinear term.  A low pass filtered velocity will have
less energy in the high wave modes, thus the nonlinear term $\ubarb
\cdot\triangledown \ub$ will generate higher wave modes at a reduced rate.  By
reducing the energy cascade, discontinuities are prevented from forming.

\begin{figure}[!ht]
\begin{center}
\begin{minipage}{0.48\linewidth} \begin{center}
  \includegraphics[width=.9\linewidth]{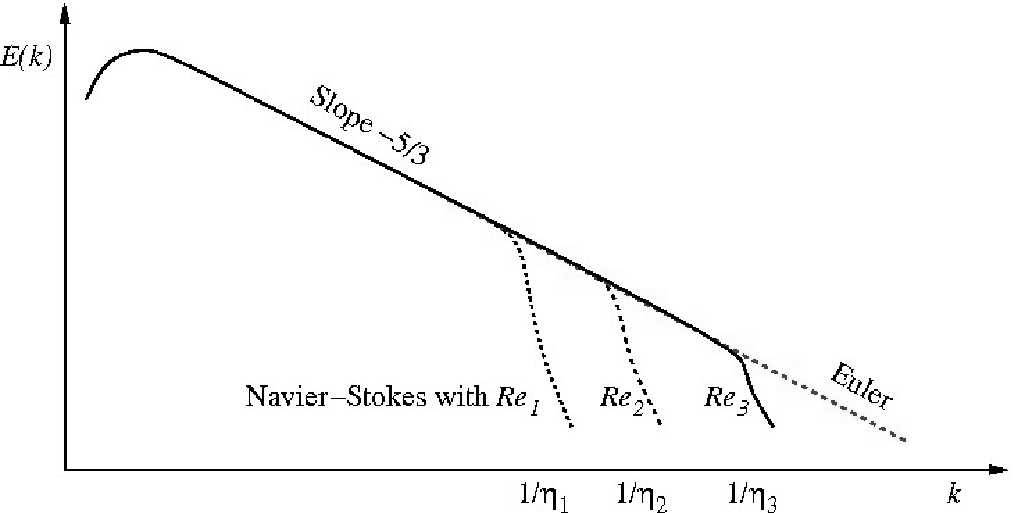}
\end{center} \end{minipage}
\begin{minipage}{0.48\linewidth} \begin{center}
  \includegraphics[width=.9\linewidth]{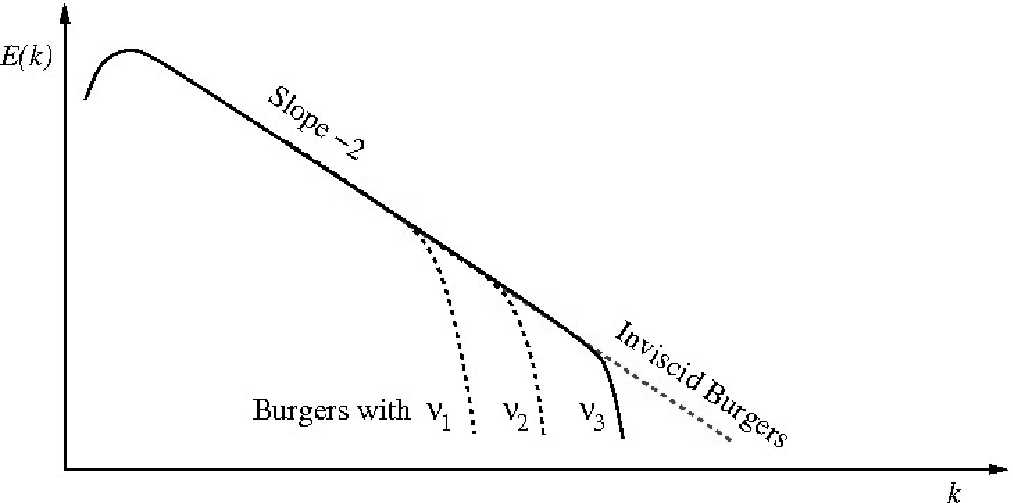}
\end{center} \end{minipage}\\
\begin{minipage}{0.48\linewidth}\begin{center} (a) \end{center} \end{minipage}
\begin{minipage}{0.48\linewidth}\begin{center} (b) \end{center}
\end{minipage}\vspace{-2mm}
\caption{Energy cascade sketches for Navier-Stokes/Euler equations and viscous/inviscid Burgers equation.  (a) Energy cascades from high wavelengths to lower wavelength at a
predicted rate of -$\frac{5}{3}$ for Navier-Stokes/Euler equations.  Kinetic
energy drops drastically upon
reaching a certain wavelength.  That wavelength $\eta$ is referred to as the
Kolmogorov scale.  Here $Re_1 < Re_2 <Re_3$. (b) For Burgers equation the energy
cascade has a slope of -$2$ until viscosity begins to exert its influence. Here
$\nu_3<\nu_2<\nu_1$.}
\label{kscale}
\end{center}
\end{figure}

\section{Filters}\label{filters}
In Equation \eref{CFBb}, the filtered velocity is obtained by convoluting the unfiltered
velocity $\ub$ with a low pass filter $g$.  This section discusses the
requirements and properties of the low pass filter.  Using physical and
analytical arguments the class of relevant filters is outlined.

\subsection{Filtering as a Weighted Average}
From a physical perspective, Equations \eref{CFBa} and \eref{CFBb} describe a
fluid where the convective velocity of a particle is the weighted average  of
the velocities of the particles around it.  In such an averaging, there are
several guidelines that seem intuitively reasonable.  The averaging should give
no particle a negative weight, should give no directional preference, and should
give more weight to particles that are physically closer.  These guidelines are
easily expressed as mathematical properties of the filter.   Thus $g$ should be
nonnegative, decreasing, and radially symmetric.  Furthermore, as a weighted
average the integral of $g$ over the domain should be precisely one.

\subsection{Characteristic Wavelength Parameter}
Each such filter $g$ is equipped with a parameter $\alpha$, such that $\alpha$
is the characteristic wavelength of that filter.  This parameter is introduced
by scaling the filter as such:
\begin{equation}
g^\alpha=\frac{1}{\alpha}g(\frac{\xb}{\alpha}).
\end{equation}
This scaling is also realized in the Fourier domain by noting that
\begin{equation}
\widehat{g^\alpha}(\kb)= \widehat{g}(\alpha \kb).
\end{equation}
Thus as $\alpha$ becomes smaller, the wavelength where the filter exerts
influence also become smaller.

With this scaling, the filter remains normalized, nonnegative, decreasing, and
isotropic.  Furthermore as $\alpha \to 0$, $g^\alpha$ approaches the Dirac delta
function, and Equation \eref{CFBa} formally limits to inviscid Burgers equation.
 This parameter is crucial in resolving features in $\ubarb$.  When convoluting
$g^\alpha$ with $\ub$, features in $\ub$ that have length less than $\alpha$
will be smoothed out.  It is also seen in Section \ref{shockthickness} that
$\alpha$ controls shock thickness.

\subsection{Fourier Decay Properties}
In the following section, Theorem \ref{existencetheorem} states that system
\eref{IVPa}, \eref{IVPb}, and \eref{IVPc} has a continously differentiable solution, if $g(\textbf{x})$ $\in$
$W^{1,1}(\mathbb{R}^n)$ or alternatively
\begin{equation}
\label{filterL1}
||g||_{L^1}<\infty,
\end{equation}
and
\begin{equation}
\label{filterdxL1}
||\frac{\p}{\p x_j}g||_{L^1}<\infty.
\end{equation}
	
Condition (\ref{filterdxL1}) has implication in the Fourier domain.  The
Riemann-Lebesgue Lemma demands that if 
$\frac{\p}{\p x_j}g\in L^1$ then $\widehat{\frac{\p}{\p x_j}g}=ik_j\hat{g}\in
C_0.$  Thus not only must $\hat{g}\in C_0$,
but the condition,

\begin{equation}
\label{filterdecay}
\lim_{|\kb| \to \infty} |\kb|\hat{g}(\kb)=0,
\end{equation}
must be satisfied in order to meet (\ref{filterdxL1}).  Note that
\eref{filterdecay} is a necessary, not sufficient, condition to meet Theorem
\ref{existencetheorem}.

\begin{table}

\begin{tabular}{lc}
\vspace{3 mm}\\
\hline
Properties & Mathematical Expression\\
\hline
Normalized & $\int g$=1 \\
Nonnegative & $g(\xb) >0,\, \forall \xb $   \\
Decreasing & $|\xb_1|\ge |\xb_2| \Rightarrow g(\xb_1) \ge g(\xb_2)$  \\
Symmetric  & $|\xb_1|= |\xb_2| \Rightarrow g(\xb_1) = g(\xb_2)$    \\
Fourier Decay  & $\lim_{|\kb| \to \infty} |\kb|\hat{g}(\kb)=0$\\
\hline
\end{tabular}
\caption{\label{FilterProperties}This table succinctly lists the requirements of
the low pass filters employed in Equation \eref{CFBb}.}
\end{table}

\section{Existence and Uniqueness Theorem}\label{ExistenceTheorem}

  It has been previously proven \cite{BhatHS:06a} that the initial value problem
\eref{IVPa}, \eref{IVPb}, and \eref{IVPc} using the Helmholtz filter
(\ref{helmholtzfiltera}) has a continuously differentiable solution for
continuously differentiable initial conditions.  Taking inspiration from that
work, the following theorem generalizes the existence and uniqueness result into
multiple dimensions and a variety of filters.

\begin{theorem}
\label{existencetheorem}
Let $g(\mathbf{x})$ $\in$ $W^{1,1}(\mathbb{R}^n)$ and $\ub_0(\mathbf{x})$ $\in$
$C^1(\mathbb{R}^n)$, then there exists a unique global solution
$\ub(\mathbf{x},t)$ $\in$ $C^1(\mathbb{R}^n)$ to the initial value problem
\eref{IVPa}, \eref{IVPb}, and \eref{IVPc}.

\label{IVP}
\begin{eqnarray}
\label{IVPa}
\ub_t+\ubarb\cdot \triangledown \ub=0\\
\label{IVPb}
\ubarb=g\ast \ub\\
\label{IVPc}
\ub(\mathbf{x},0)=\ub_0(\mathbf{x})
\end{eqnarray}

\end{theorem}

\begin{proof}  Begin by shifting perspective into the material view.  Associate a map
$\phib(\boldsymbol{\xi},t): \mathbb{R}^n \to \mathbb{R}^n$ as the map from
a particle's original position, ($\boldsymbol{\xi}$), to the particle's  position at time $t$, ($\xb$).
Naturally this dictates that $\phib(\boldsymbol{\xi},0)=\boldsymbol{\xi}.$ 
Associate this material map with the velocity $\ubarb$ by

\begin{equation}
\frac{\partial}{\partial
t}\phib(\boldsymbol{\xi},t)=\ubarb(\phib(\boldsymbol{\xi},t),t).
\end{equation}

It can then be seen that \eref{IVPa} simply becomes

\begin{equation}
\frac{\partial}{\partial t} \ub (\phib(\boldsymbol{\xi},t),t)=0,
\end{equation}
and thus, $\ub (\phib(\boldsymbol{\xi},t),t)=\ub
(\phib(\boldsymbol{\xi},0),0)=\ub_0(\boldsymbol{\xi})$, which implies

\begin{equation}
\label{linfinitybound}
\left|\left| \ub \right|\right|_{L^\infty}=\left|\left| \ub_0
\right|\right|_{L^\infty}.
\end{equation}

If it is assumed that $\phib(\boldsymbol{\xi},t)$ has a continuously
differentiable inverse $\phib^{-1}(\xb,t)$, then \eref{IVPa}, \eref{IVPb},
and \eref{IVPc} will have the continuously differentiable solution

\begin{equation}
\label{solution}
\ub (\mathbf{x},t)=\ub_0 (\phib^{-1}(\mathbf{x},t)).
\end{equation}

A sufficient condition for such an  $\phib^{-1}$ to uniquely exist is the
Jacobian of $\phib$ to be non-zero for all positions and time. Thus if
$J(\phib(\boldsymbol{\xi})) \neq 0, \forall \, \boldsymbol{\xi}, t$, then
\eref{IVPa}, \eref{IVPb}, and \eref{IVPc} is uniquely solved by \eref{solution}.

Using a result from Aris \cite{ArisR:62a}, the time derivative of the Jacobian
is found to be
\begin{equation}
\frac{\partial}{\partial t} J =\triangledown \cdot \ubarb \, J,
\end{equation}
which is a differentiable equation with solution 
\begin{equation}
\label{jacobiansolution}
J=J_0 \exp(\int_0^t \triangledown \cdot \ubarb \,dt) 
\end{equation}
Clearly $J_0=1$, since
$\phib(\boldsymbol\xi,0)=\boldsymbol\xi$.  Thus it is clear that $J$
remains non-zero if $\left|\int_0^t \triangledown \cdot \ubarb \,dt\right| <
\infty$.

Given that $g(\mathbf{x})$ $\in$ $W^{1,1}(\mathbb{R}^n)$, there exists $M \in
\mathbb{R}$, such that
\begin{equation}
\label{l1bound}
\left|\left| \frac{\partial}{\partial{x_i} } g \right|\right|_{L^1} \leq
M<\infty,\qquad \forall \, i.
\end{equation}
Knowing that $\frac{\partial}{\partial{x_i} } g \in L^1$ and $\ub\in L^\infty$,
we know that $\frac{\partial}{\partial{x_i} } \ubar_j$ exists and that
\begin{equation}
\frac{\partial}{\partial{x_i} } \ubar_j =\frac{\partial}{\partial{x_i} } g \ast
u_j.
\end{equation}
Using this and bounds \eref{linfinitybound} and \eref{l1bound}, Young's
inequality states
\begin{equation}
\left|\left| \frac{\partial}{\partial{x_i} } \ubar_j \right|\right|_{L^\infty}
\leq \left|\left| \frac{\partial}{\partial{x_i} } g \right|\right|_{L^1}
\left|\left| u_j \right|\right|_{L^\infty} \leq M \left|\left| u_{0_j}
\right|\right|_{L^\infty},
\end{equation} 
which leads directly to the bound
\begin{equation}
\left|\left| \int_0 ^t \triangledown \cdot \ubarb \,dt \right|\right|_{L^\infty}
\leq nM \max_j \left|\left| u_{0_j} \right|\right|_{L^\infty}t.
\end{equation}

Thus for finite time, the Jacobian of $\phib$ remains uniquely invertible, with a
continuously differentiable inverse and thus \eref{solution} is a unique
$C^1(\mathbb{R}^n)$ solution to \eref{IVPa}, \eref{IVPb}, and \eref{IVPc}.

\end{proof}

\begin{theorem}
\label{existencetheorem2}
Let $g(\mathbf{x})$ $\in$ $W^{1,1}(\mathbb{R}^n)$ and $\ub_0(\mathbf{x})$ $\in$
$L^\infty(\mathbb{R}^n)$, then there exists a unique global solution
$\ub(\mathbf{x},t)$ $\in$ $L^\infty(\mathbb{R}^n)$ to the initial value problem
\eref{IVPa}, \eref{IVPb}, and \eref{IVPc}.
\end{theorem}

\begin{proof}
The proof is the same as for Theorem \ref{existencetheorem}. 
$\phib(\boldsymbol{\xi},t)$ still has a unique continuously differentiable
inverse, and the solution remains $\ub (\mathbf{x},t)=\ub_0
(\phib^{-1}(\mathbf{x},t)),$  but now lacks continuity due to the initial
conditions.
\end{proof}

\section{Invariants of Motion}\label{invariantsofmotion}
Inviscid and viscous Burgers equation have invariants of motion that are a
result of their inherent geometric structure.  Often these invariants of motion
can only be realized with additional constraints which impart some physical meaning.   These constraints are
specifically: restricition to a single dimension, assuming the velocity is the
gradient of a potential function, $\ub=\triangledown \phi$, and by the addition
of a continuity equation. Both restricting to a single dimension and assuming a potential function establishes a curl free velocity.  Introducing a continuity equation similar to that in the Euler equations introduces density. Table \ref{invariants} shows the various invariants of motions for the equations under different constraints.

By restricting inviscid Burgers equation to one dimension there are a countably
infinite number of conserved quantities \cite{DavoudiJ:01a}.  When viscous
Burgers equation is restricted to one dimension, only the wave mass is
conserved.  Similarly, CFB conserves wave mass when restricted to one dimension
\cite{Holm:03a}.  This result is dependent upon the filter $g$ being an even
function, which was a criteria set down in section \ref{filters}. In the special
case of the the Helmholtz filter, this quantity is also a Hamiltonian structure
\cite{Holm:03a, BhatHS:06a}. 

One dimensional CFB has an additional conserved quantity, total variation, that is not found in either
inviscid or viscous Burgers equation.  CFB can be considered a convection
equation, with convective velocity $\ubar.$  Since the solution is a remapping
of the initial condition, the total variation remains constant over time.  For
continuous initial conditions and thus continuous solutions, this can be
verified directly by noting that $T.V.(f)=\int|f_x| dx$  for continuous
functions. By taking the derivative of \eref{CFB1Da} and integrating its
absolute value, one obtains
\begin{equation}
\frac{\p}{\p t} \int |u_x|\, dx + \int sgn(u_x) (\ubar u_x)_x\,dx=0. 
\end{equation}
Break the second term into intervals where $sgn(u_x)$ remains constant. $u_x$
and $\ubar$ are continuous due to Theorem \ref{existencetheorem}, so where
$sgn(u_x)$ changes sign, $u_x=0$.  Thus the second term is zero and $\int
|u_x|\, dx$ remains constant over time. This was established by Bhat and Fetecau for the Helmholtz filters \cite{BhatHS:06a}, but remains true for any filter satisfying Theorem \ref{existencetheorem}.

When assuming a potential function for the velocity, $\ub =\triangledown \phi$, 
in inviscid and viscous Burgers equation wave mass is again conserved.  Notice
that restricition to one dimension trivially implies a potential function for
the velocity. Similarly CFB conserves wave mass with the assumption of a
potential function.  The velocities $\ub$ and $\ubarb$ have different potential
functions related through the filter $g$.  Again the conservation of wave mass
is dependent on $g$ being radially symmetric.

The addition of a continuity equation to inviscid Burgers equation results in
multiple invariants of motion \cite{DavoudiJ:01a}.  It can be easily shown that
this addition conserves mass, momentum, and kinetic energy.  CFB also conserves mass, momentum, and kinetic energy when a
continuity equation involving the filtered velocity is added.

CFB has another quantity that remains constant over time.  As can be seen in the
existence and uniqueness theorem in section \ref{ExistenceTheorem}, the infinity
norm of the unfiltered velocity remains constant.  Indeed both the maximum and
minimum of the initial condition is preserved over time.  This was shown for the
one dimensional case by Holm and Staley \cite{Holm:03a}.

It should be noted that almost all of the invariants of motion are found using the unfiltered velocity.  The filtered velocity, which is to be considered the physical quantity, has only wave mass preserved.  In order for the filtered velocity to accurately model the macroscopic (low wavenumber) behavior, energy must be leaving the system in the high wavenumber spectrum.  Indeed, in numerical simulation in sections \ref{spectralenergy} and \ref{energynorms} show that energy is leaving the filtered velocity at the high wavenumbers.

Besides wave mass, the filtered velocity does not appear in the invariants of motion.  It was mentioned in the introduction that the filtered velocity is meant as the physically relevant quantity and is to capture the low wavenumber behavior.  Thus it should not capture all of the energy behavior in the higher wave modes.  The unfiltered velocity is intended to capture such behavior.  Therefore the invariants of motion for the unfiltered velocity would be more properly compared to those of inviscid Burgers equations, just as the invariants of motion for the filtered velocity should be compared to those of viscous Burgers equation. It should be noticed that inviscid Burgers equation can be thought of as a conservation law for the wave mass.  Thus this invariant of motion is of primary importance and is noticeably preserved in viscous Burgers equations and in both velocities in convectively filtered Burgers equation.

\begin{table}\label{invariants}
\begin{tabular}{lll}

\centre{3}{Viscous Burgers}
\vspace{2mm}\\

Additional Constraints $\,\,\,$ & System Equations & Invariants of Motion\\
\hline
One Dimensional  & $\,\,\,u_t+u u_x=\nu u_{xx}$  & $\int u$ \\
\hline
Potential Function & 
			$\begin{array}{l}
			\ub_t+\ub \cdot \triangledown \ub = \nu \bigtriangleup
\ub, \\ 
			\ub=\bigtriangledown \phi
			\end{array}$
										
& $\int u_i$\\

\hline
\vspace{5mm}\\ 
\centre{3}{Inviscid Burgers}
\vspace{2mm}\\

Additional Constraints $\,\,\,$& System Equations & Invariants of Motion\\
\hline
One Dimensional  & $\,\,\,u_t+u u_x=0$  & $\int u^n, \forall n \in \mathbb{Z}$
\\
\hline
Potential Function & 
			$\begin{array}{l}
			\ub_t+\ub \cdot \triangledown \ub = 0, \\ 
			\ub=\bigtriangledown \phi
			\end{array}$
									& $\int
u_i$\\
\hline
Continuity Equation & 
			$\begin{array}{l}
			\ub_t+\ub \cdot \triangledown \ub = 0, \\ 
			\rho_t+\triangledown (\rho \ub)=0
			\end{array}$
									& $\int
\rho, \,\, \int \rho u_i, \,\, \int \rho \ub \cdot \ub$\\
\hline
\vspace{5mm}\\ 
\centre{3}{Convectively Filtered Burgers}
\vspace{2mm}\\

Additional Constraints $\,\,\,$& System Equations & Invariants of Motion\\
\hline
One Dimensional  & $\,\,\,u_t+\ubar u_x=0$  & $\int u,\,\,\int
\ubar,\,\,TV(u),\,\,||u||_\infty$ \\
\hline
Potential Function & 
			$\begin{array}{l}
			\ub_t+\ubarb \cdot \triangledown \ub = 0, \\ 
			
			\ub=\bigtriangledown \phi
			\end{array}$
									& $\int
u_i, \,\, \int \ubar_i,\,\,||\ub||_\infty$ \\
\hline
Continuity Equation & 
			$\begin{array}{l}
			\ub_t+\ubarb \cdot \triangledown \ub = 0, \\ 
			\rho_t+\triangledown (\rho \ubarb)=0
			\end{array}$
									& $\int
\rho, \,\, \int \rho u_i, \,\, \int \rho \ub \cdot \ub,\,\,||\ub||_\infty$\\
\hline

\end{tabular}
\caption{This table lists the additional constraints, the
modified equations, and the corresponding invariants of motion for viscous,
inviscid, and convectively filtered Burgers equation.}
\end{table}

\section{Traveling Wave Solutions}\label{travelingwavesolution}
In this section, one dimensional CFB, \eref{CFB1Da} and \eref{CFB1Db}, is shown
to have a traveling wave solution. For the Helmholtz filter, these have already been found by Bhat and Fetecau \cite{BhatHS:06a} and Holm and Staley \cite{Holm:03a}.  This section generalizes the results to general filters.  This traveling wave solution is a moving
shock that handles an amplitude drop or increase.  It satisfies the limiting
boundary conditions
\begin{equation}
\lim_{x\to \infty} u(x,t)=u_r,\qquad \lim_{x\to -\infty} u(x,t)=u_l.
\end{equation}

Inviscid Burgers is known to have weak solutions
\begin{equation}\label{jumpsolution}
u(x,t)=\left\{\begin{array}{ll}
u_l & x < ct \\ 
u_r & x\geq ct
\end{array}\right.
\end{equation}
with $c=\frac{1}{2}(u_r+u_l).$  The speed at which discontinuities travel is
dictated by the Rankine-Hugoniot jump condition \cite{LaxPD:73a}. 

It is not difficult to show that \eref{jumpsolution} is a weak solution to one
dimensional CFB as well.  With Equation \eref{jumpsolution} as the unfiltered
velocity solution, the filtered velocity is
\begin{equation}\label{filteredjumpsolution}
\ubar(x,t)=(u_r-u_l)\int^{x-ct}_{-\infty}g^\alpha(s)\,ds +u_l,
\end{equation}
with $g^\alpha$ as the filter and $c=\frac{1}{2}(u_r+u_l)$.  As filter $g$ is an
even function, the value of the filtered velocity in the middle of the shock is
the speed of the shock, $\ubar(ct,t)=c.$  This is crucial in verifying that
\eref{jumpsolution} is a weak solution.

There are some notable properties of the traveling wave solution.  It travels at
the same speed as traveling wave solutions to inviscid Burgers.  As it is
desirable for CFB to capture wave speed accurately, this is promising. 
Additionally, as $\alpha \to 0$, the filter $g$ approaches the Dirac delta
function, thus \eref{filteredjumpsolution} converges to \eref{jumpsolution}, thus the traveling solutions to CFB converges to weak solutions of inviscid Burgers equation.

It is also noted that when $u_l< u_r$ that the traveling wave solution will converge to a discontinuity, where $u(ct^-,t)<u(ct^+,t)$.  The Lax entropy condition states is that at points of discontinuity $u(x^-,t) > u(x^+,t)$.  Thus the traveling wave solutions to CFB will converge to a weak solution of inviscid Burgers equation which violates the entropy condition.  This non-entropic behavior can be avoided by only allowing continuous initial conditions.  This is shown in \cite{NorgardG:08a}.

\section{Numerics}\label{numerics}
For numerical simulations a pseudospectral method with a third order total variation diminishing (TVD) Runge-Kutta scheme was utilized.  Spatial derivatives and filtering were performed in the Fourier space.  The equations were advanced in time using the optimal third order TVD Runge-Kutta scheme presented by Gottlieb and Shu \cite{GottliebS:98a}.  The CFL coefficient was chosen at $c=0.3$, with runs typically conducted at resolutions of $2^{12}=4096$ for one dimension and $256\times256$ for two dimensions. Aliasing errors were handled using the same technique as in Holm and Staley \cite{Holm:03a}.

It should be noted that this numerical method captures the behavior in both the unfiltered and filtered velocity.  The filtered velocity is considered the physically relevant quantity and requires less resolution to capture, thus a more efficient numerical method would resolve only the filtered velocity using Equation \eref{smoothonly}.  However, since the priority currently lies in exploring the properties of CFB rather than in making efficient simulations, the scheme was designed to capture both velocities.

Our groups has extensive experience in using such techniques with LAE-$\alpha$ and LANS-$\alpha$ simulations where the accuracy and stability of the approach has been established \cite{Mohseni:03a}.  Of particular interest, we have shown that the $H^1$ norm of the energy in LAE-$\alpha$ computations are numerically conserved in longtime calculation \cite{Mohseni:03a}.

\section{Shock Behavior}\label{shockbehaviour}

\subsection{Unfiltered Velocity}\label{unfilteredvelocity}
In section \ref{ExistenceTheorem}, it was proven that the unfiltered velocity has a
continuous derivative for all time.  This was shown by proving that the Jacobian of the material map remains nonzero.  In Equation \ref{jacobiansolution} the Jacobian is shown to have an exponential structure. As it happens the Jacobian remains nonzero, but approaches zero rapidly.  From a method of characteristics point of view, the characteristics are growing continuously closer, but never intersect.  Thus the unfiltered velocity will form shocks of continually smaller thickness.  In numerical simulations, this process will inevitably drop below the finite resolution, thus from a numerical perspective, the unfiltered velocity becomes discontinuous.  The filtered velocity, however, appears continuous in the numerical simulations.  The shocks in the filtered velocity will not drop below the resolution of the numerical run provided a large enough $\alpha$ has been chosen. For the average run, $\alpha$ is chosen to be at least a magnitude greater than $\Delta x$.  This puts approximately 15 points across the width of the shock, which seems a reasonable number to resolve its features.

\subsection{Shock Profile}\label{shockprofile}

In 1D viscous Burgers equation, a wave will propagate to the right, with the
higher velocities overcoming the lower velocities, creating steep gradients. 
This steepening continues until a balance with dissipation is reached. 
Figure \ref{profileviscous} shows the evolution of viscous Burgers equation with
a Gaussian distribution as an initial condition.

\begin{figure}[!ht]
\begin{center}
\includegraphics[width=0.9\linewidth]{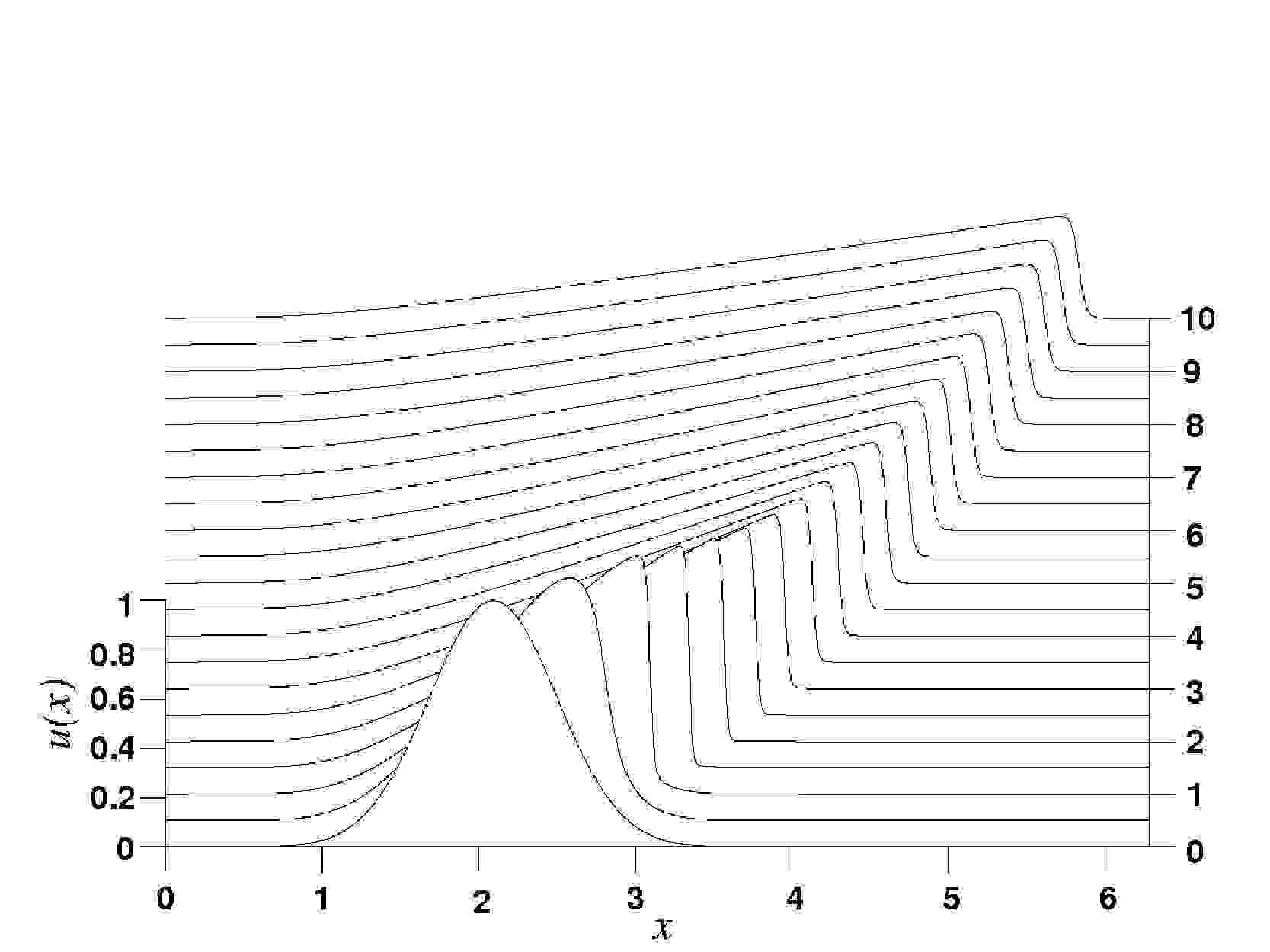}
\caption{Evolution of a Gaussian distribution under viscous Burgers equation for
$t=0, \dots ,10$.  The viscosity is $\nu=0.005$.}
\label{profileviscous}
\end{center}
\end{figure}

 In figure \ref{filtershock}, the evolution of CFB can be seen for two filters.
The behavior is qualitatively the same as viscous Burgers equation in reference
to wave propagation and shock formation.  Every filter that has been numerically
simulated has shown similar behavior.

\begin{figure}[!ht]
\begin{center}
\begin{minipage}{0.48\linewidth} \begin{center}
  \includegraphics[width=.9\linewidth]{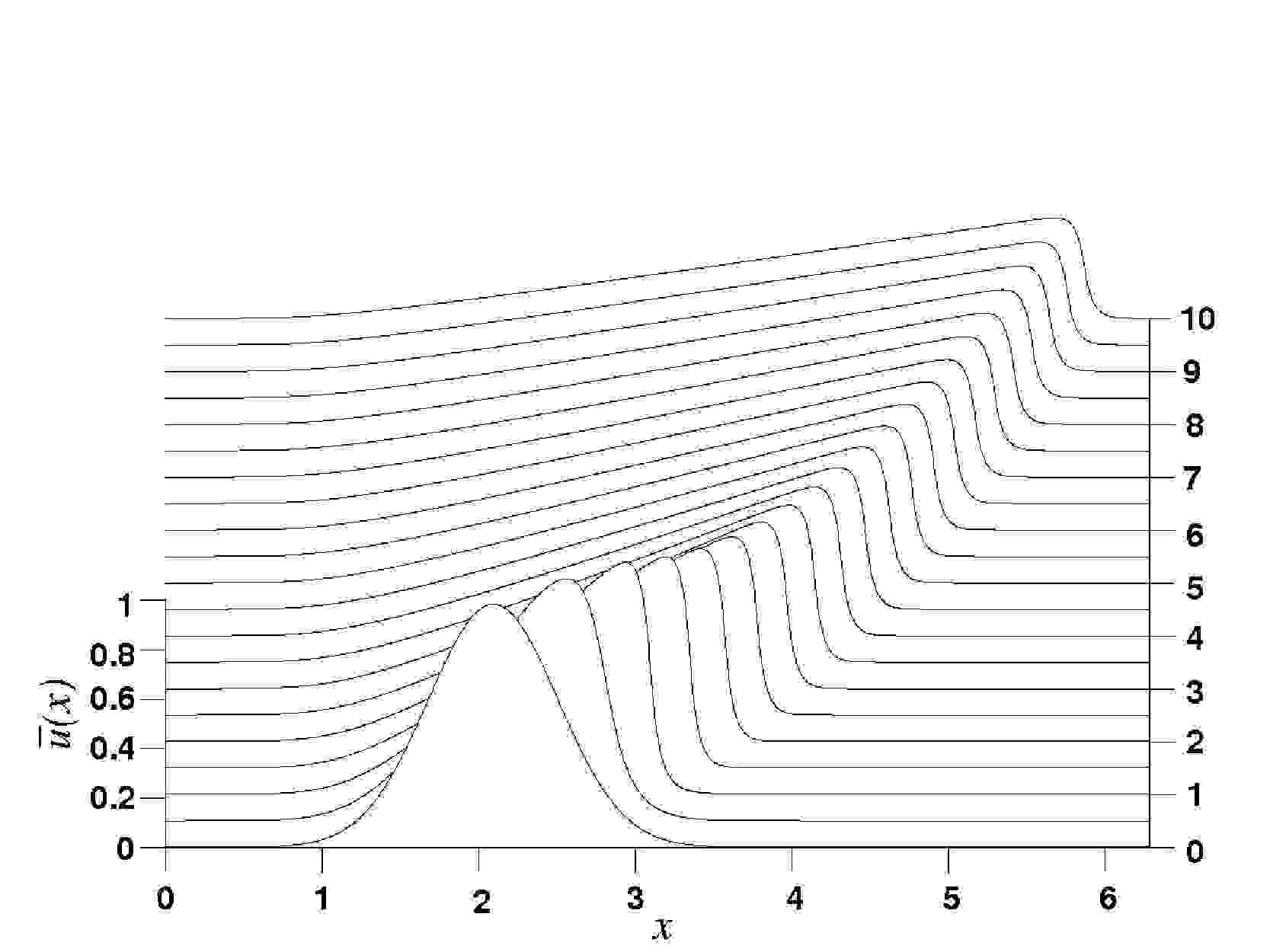}
\end{center} \end{minipage}
\begin{minipage}{0.48\linewidth} \begin{center}
  \includegraphics[width=.9\linewidth]{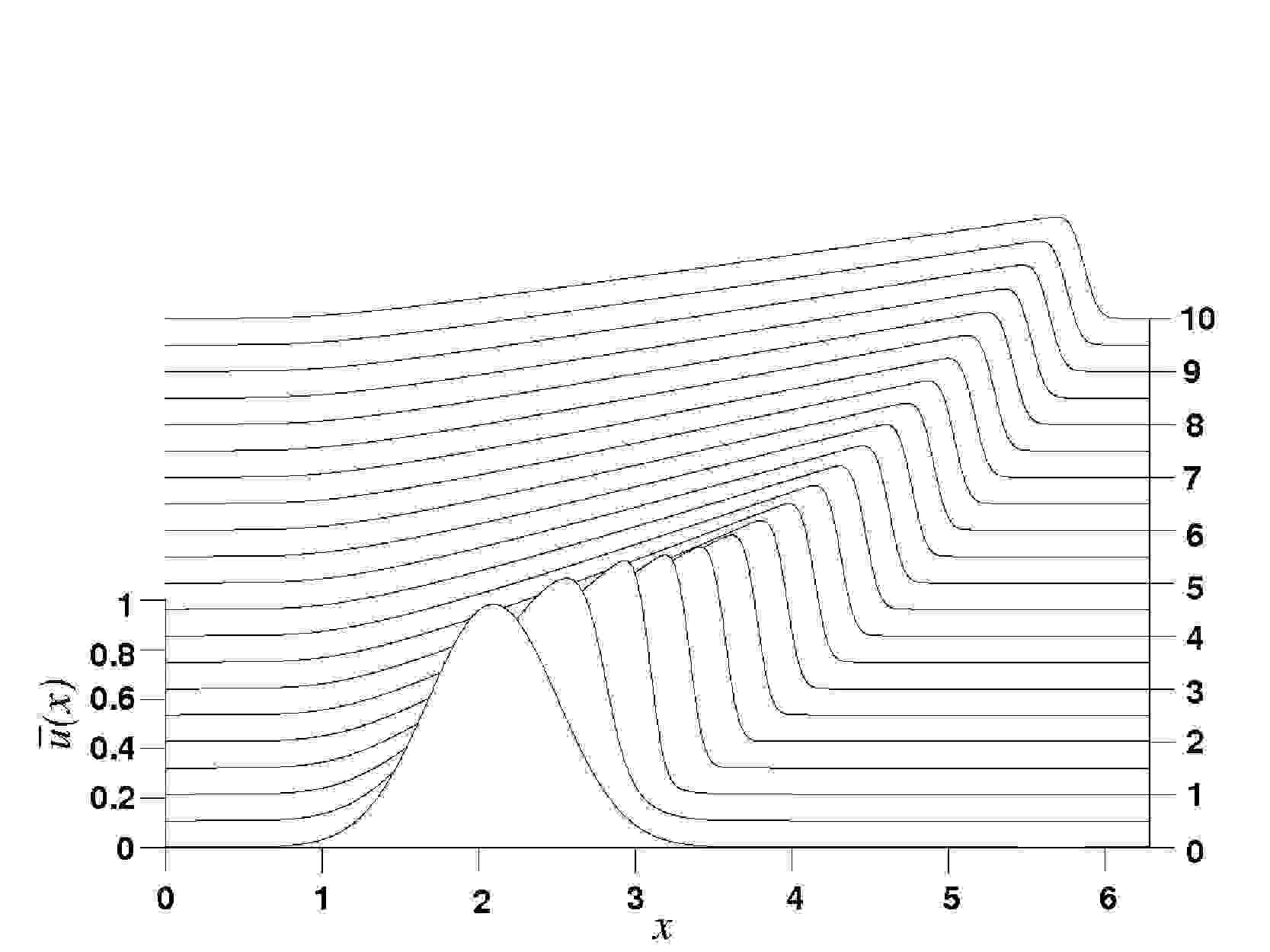}
\end{center} \end{minipage}\\
\begin{minipage}{0.48\linewidth}\begin{center} (a) \end{center} \end{minipage}
\begin{minipage}{0.48\linewidth}\begin{center} (b) \end{center} \end{minipage}
\caption{Evolution of CFB using different filters from $t=0,\dots,10$.  Only the
filtered velocity is shown.  For all filters $\alpha=0.05.$  (a) Helmholtz
filter $g(x)=\frac{1}{2}exp(-|x|)$. (b) Gaussian filter,
$g(x)=\pi^{-1/2}exp(-x^2)$. }
\label{filtershock}
\end{center}
\end{figure}

 Numerical simulations were also conducted in two dimensions.  It is more
difficult showing shock formation in 2D, but figures \ref{2Dhelmshock} and
\ref{2DGaussshock} show the evolution of a pulse under the Helmholtz and
Gaussian filter. In both runs, the pulse moves up and to the right, becoming
steeper, but never forming a discontinuity.  Once the shock has formed, the
pulse begins to decrease in amplitude.

\begin{figure}[!ht]
\begin{center}
\begin{minipage}{0.29\linewidth} \begin{center}
  \includegraphics[width=\linewidth]{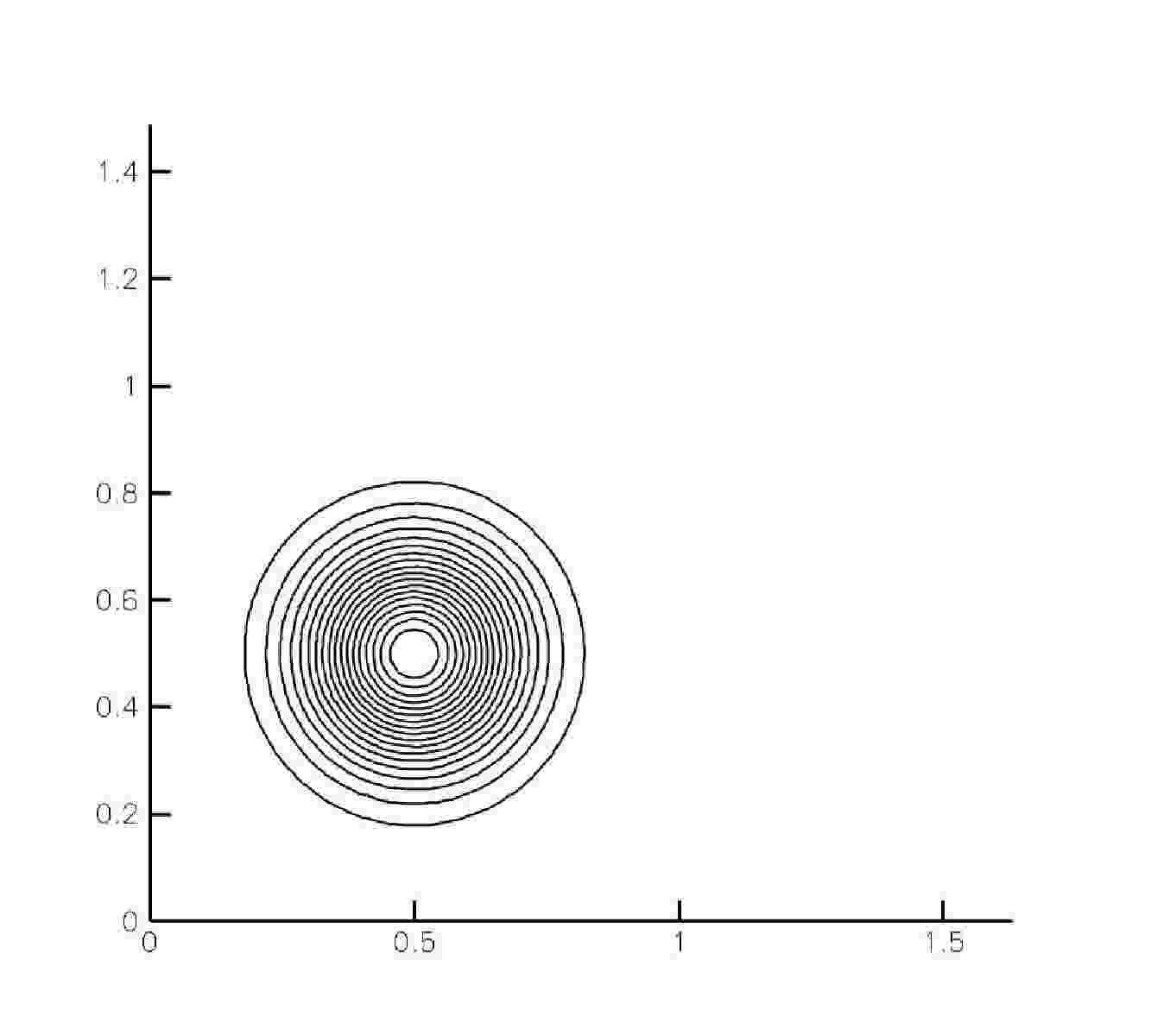}
\end{center} \end{minipage}
\begin{minipage}{0.29\linewidth} \begin{center}
  \includegraphics[width=\linewidth]{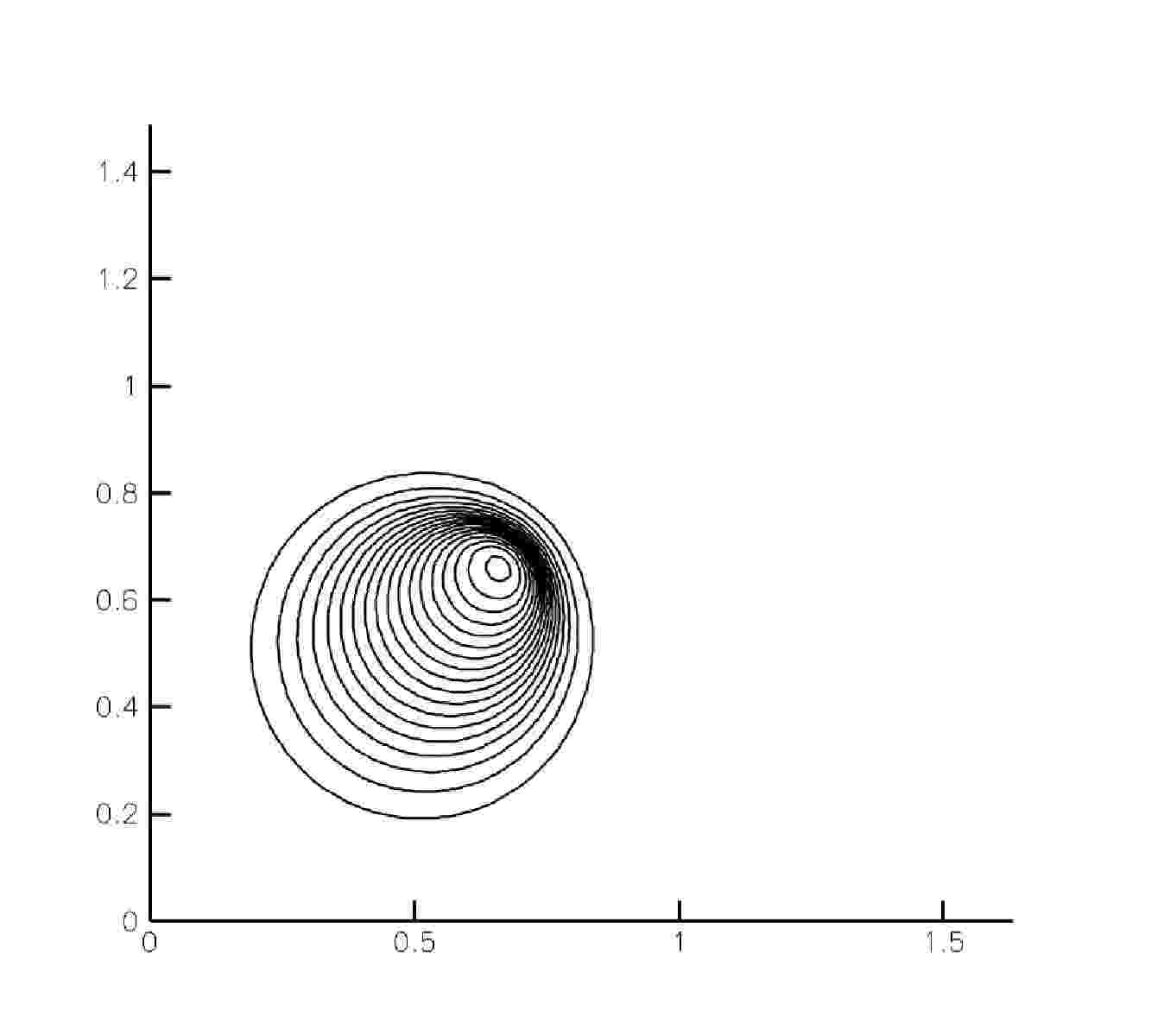}
\end{center} \end{minipage}
\begin{minipage}{0.29\linewidth} \begin{center}
  \includegraphics[width=\linewidth]{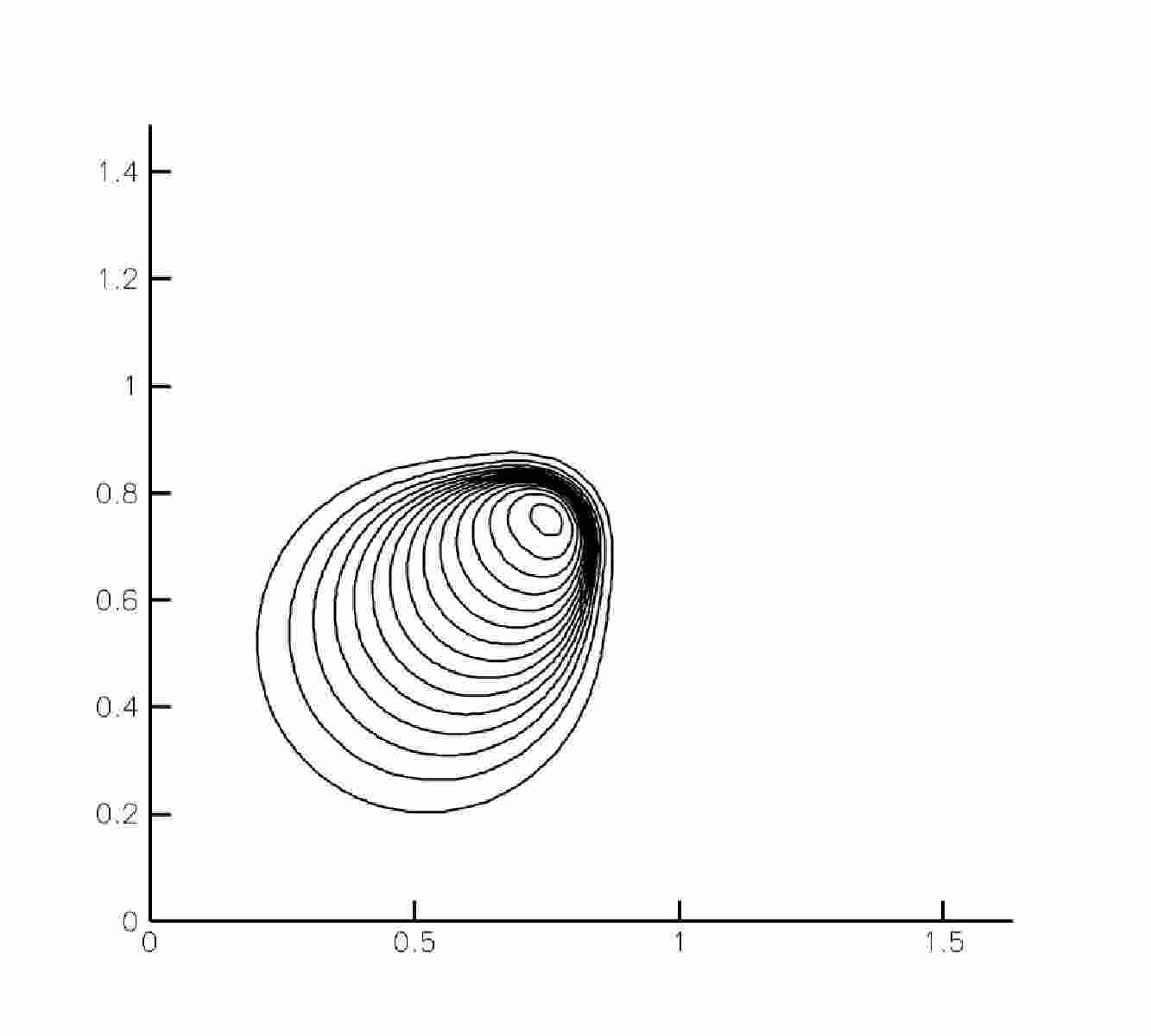}
\end{center} \end{minipage}\\
\begin{minipage}{0.29\linewidth}\begin{center} (a) \end{center}\end{minipage} 
\begin{minipage}{0.29\linewidth}\begin{center} (b) \end{center}\end{minipage} 
\begin{minipage}{0.29\linewidth}\begin{center} (c) \end{center}\end{minipage}\\
\begin{minipage}{0.29\linewidth} \begin{center}
  \includegraphics[width=\linewidth]{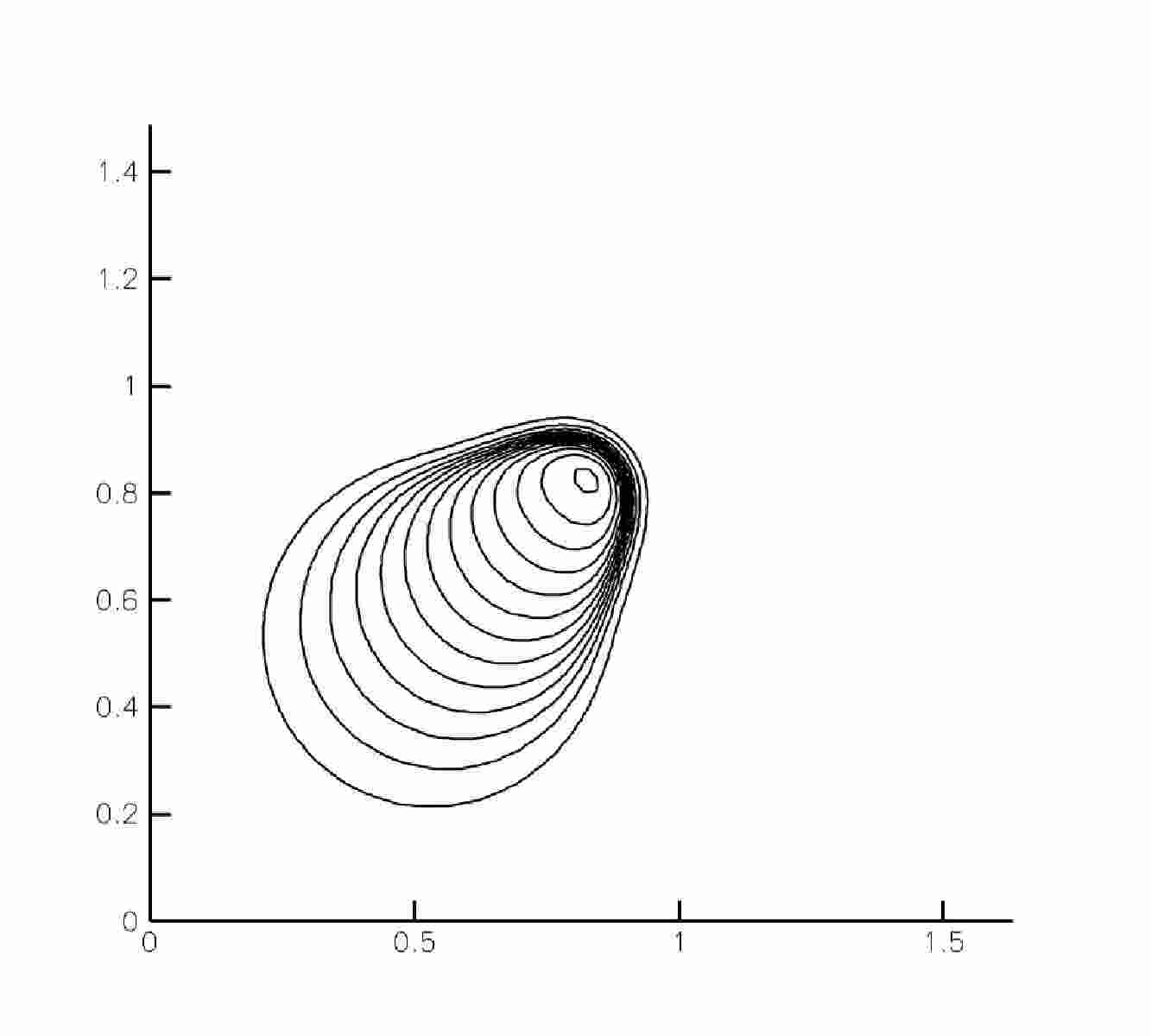}
\end{center} \end{minipage}
\begin{minipage}{0.29\linewidth} \begin{center}
  \includegraphics[width=\linewidth]{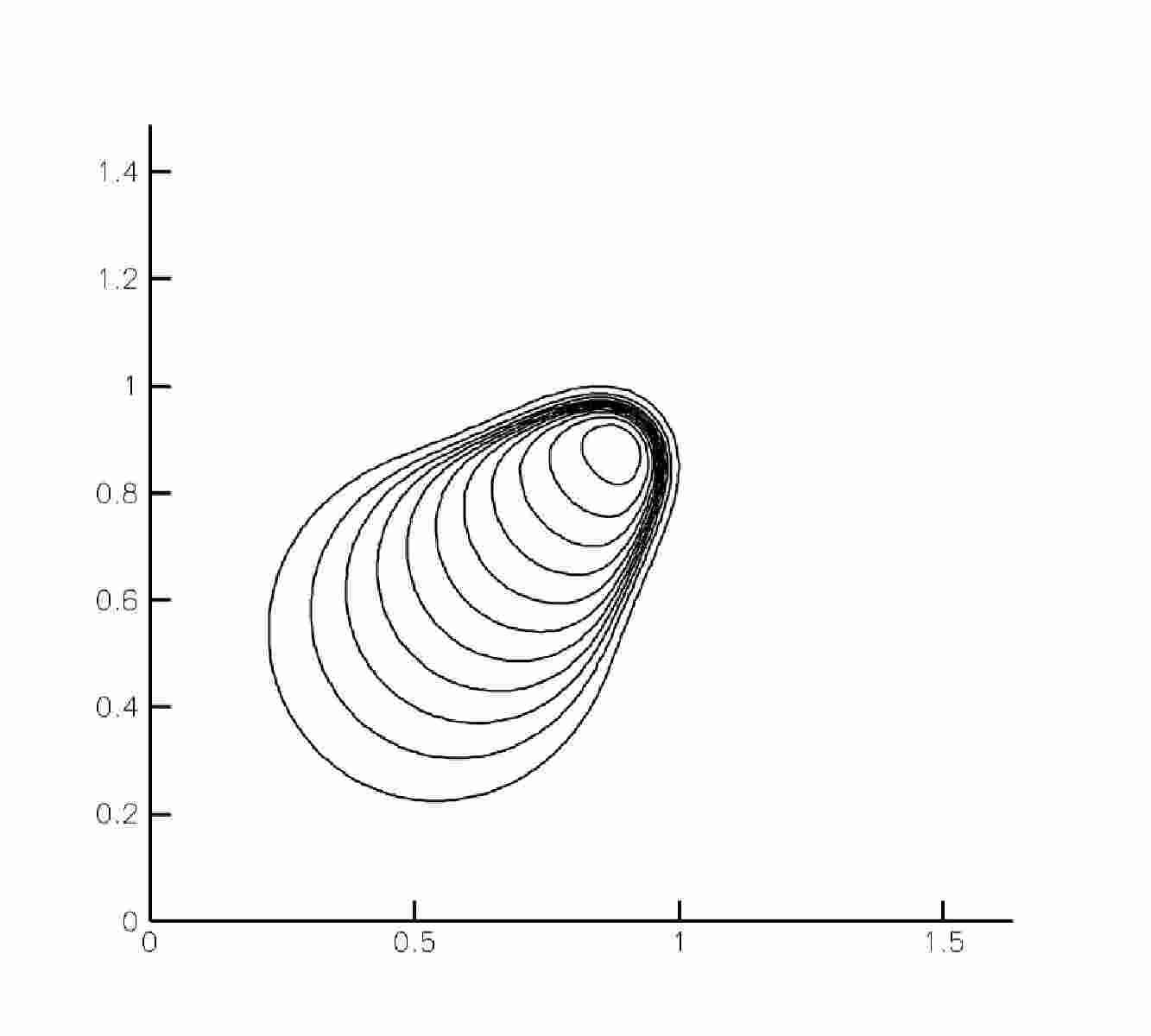}
\end{center} \end{minipage}
\begin{minipage}{0.29\linewidth} \begin{center}
  \includegraphics[width=\linewidth]{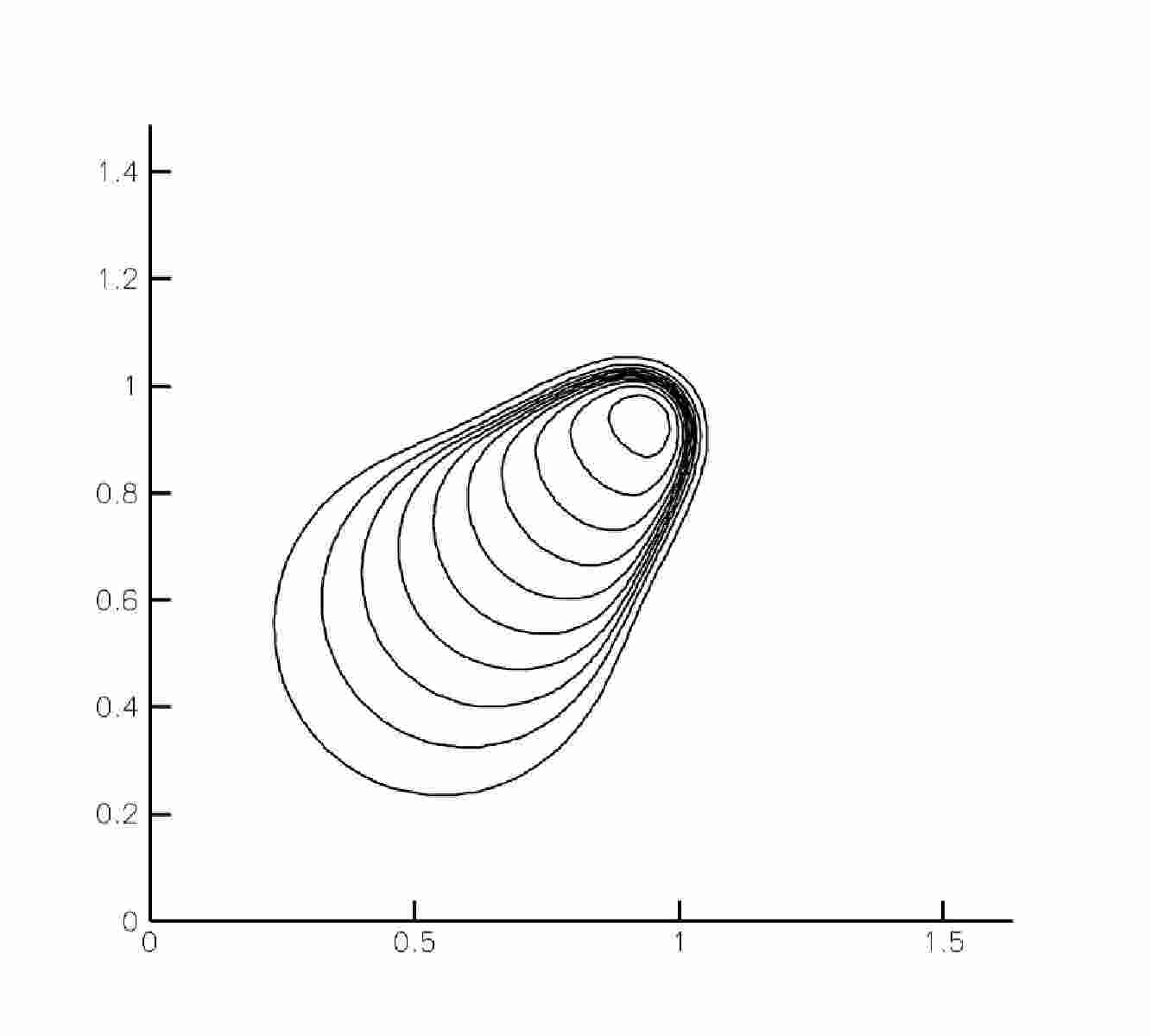}
\end{center} \end{minipage}\\
\begin{minipage}{0.29\linewidth}\begin{center} (d) \end{center}\end{minipage} 
\begin{minipage}{0.29\linewidth}\begin{center} (e) \end{center}\end{minipage} 
\begin{minipage}{0.29\linewidth}\begin{center} (f) \end{center}\end{minipage}
\caption{Contour plots showing the evolution of 2D CFB using a Helmholtz filter with a Gaussian pulse as the
initial condition. $\alpha=0.08$ with 128 $\times$ 128 grid points.}
\label{2Dhelmshock}
\end{center}
\end{figure}

\begin{figure}[!ht]
\begin{center}
\begin{minipage}{0.29\linewidth} \begin{center}
  \includegraphics[width=\linewidth]{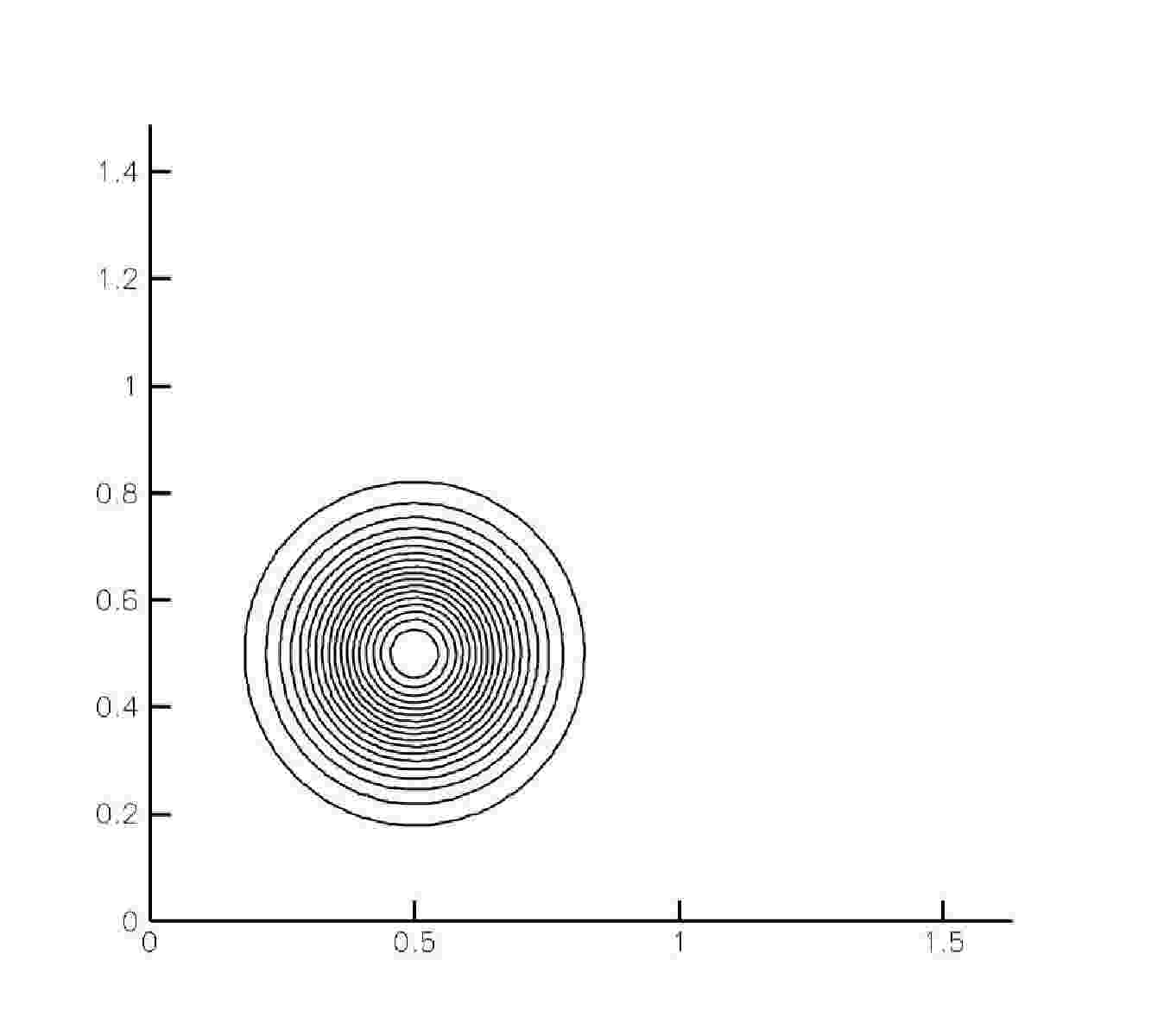}
\end{center} \end{minipage}
\begin{minipage}{0.29\linewidth} \begin{center}
  \includegraphics[width=\linewidth]{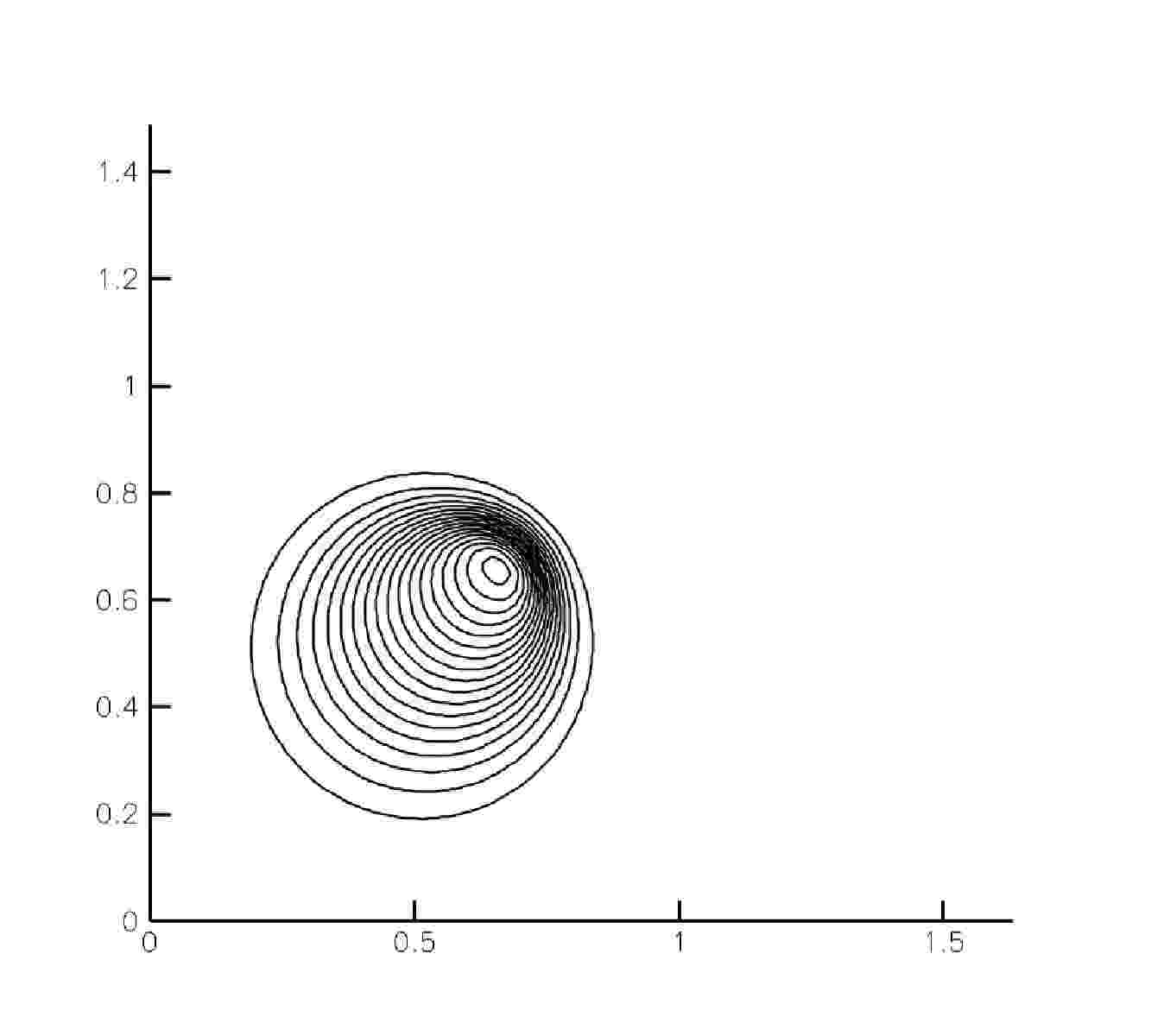}
\end{center} \end{minipage}
\begin{minipage}{0.29\linewidth} \begin{center}
  \includegraphics[width=\linewidth]{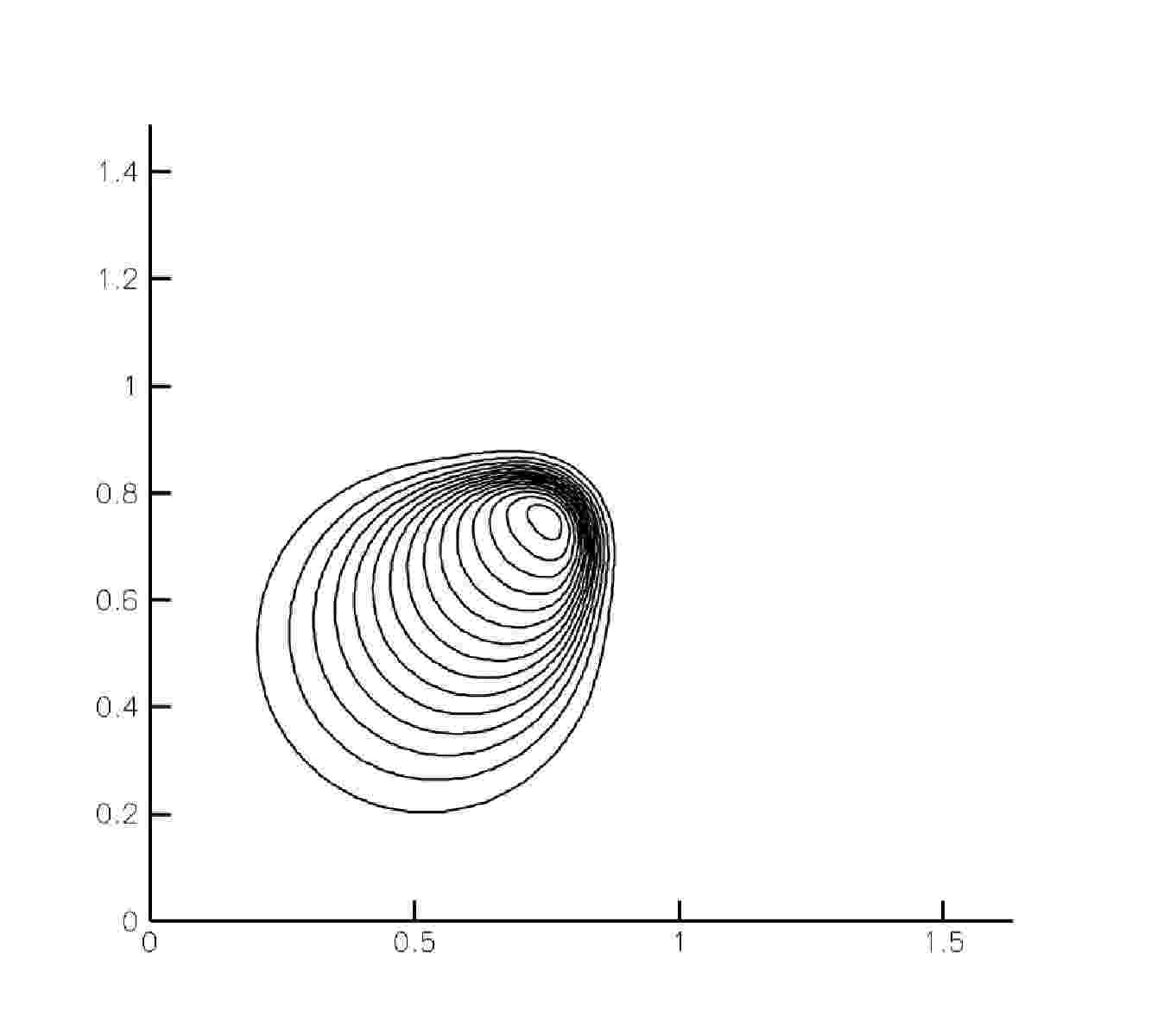}
\end{center} \end{minipage}\\
\begin{minipage}{0.29\linewidth}\begin{center} (a) \end{center}\end{minipage} 
\begin{minipage}{0.29\linewidth}\begin{center} (b) \end{center}\end{minipage} 
\begin{minipage}{0.29\linewidth}\begin{center} (c) \end{center}\end{minipage}\\
\begin{minipage}{0.29\linewidth} \begin{center}
  \includegraphics[width=\linewidth]{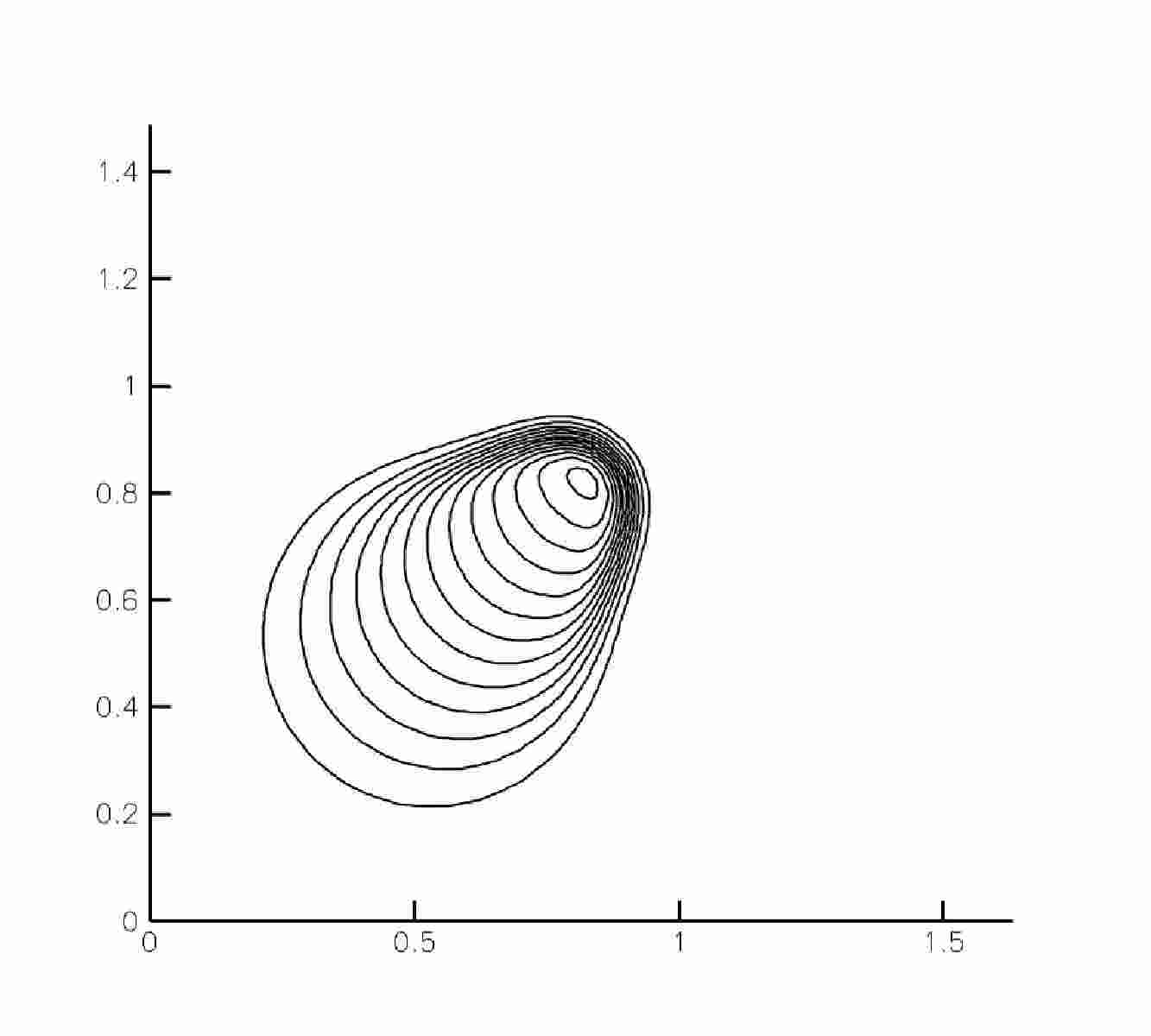}
\end{center} \end{minipage}
\begin{minipage}{0.29\linewidth} \begin{center}
  \includegraphics[width=\linewidth]{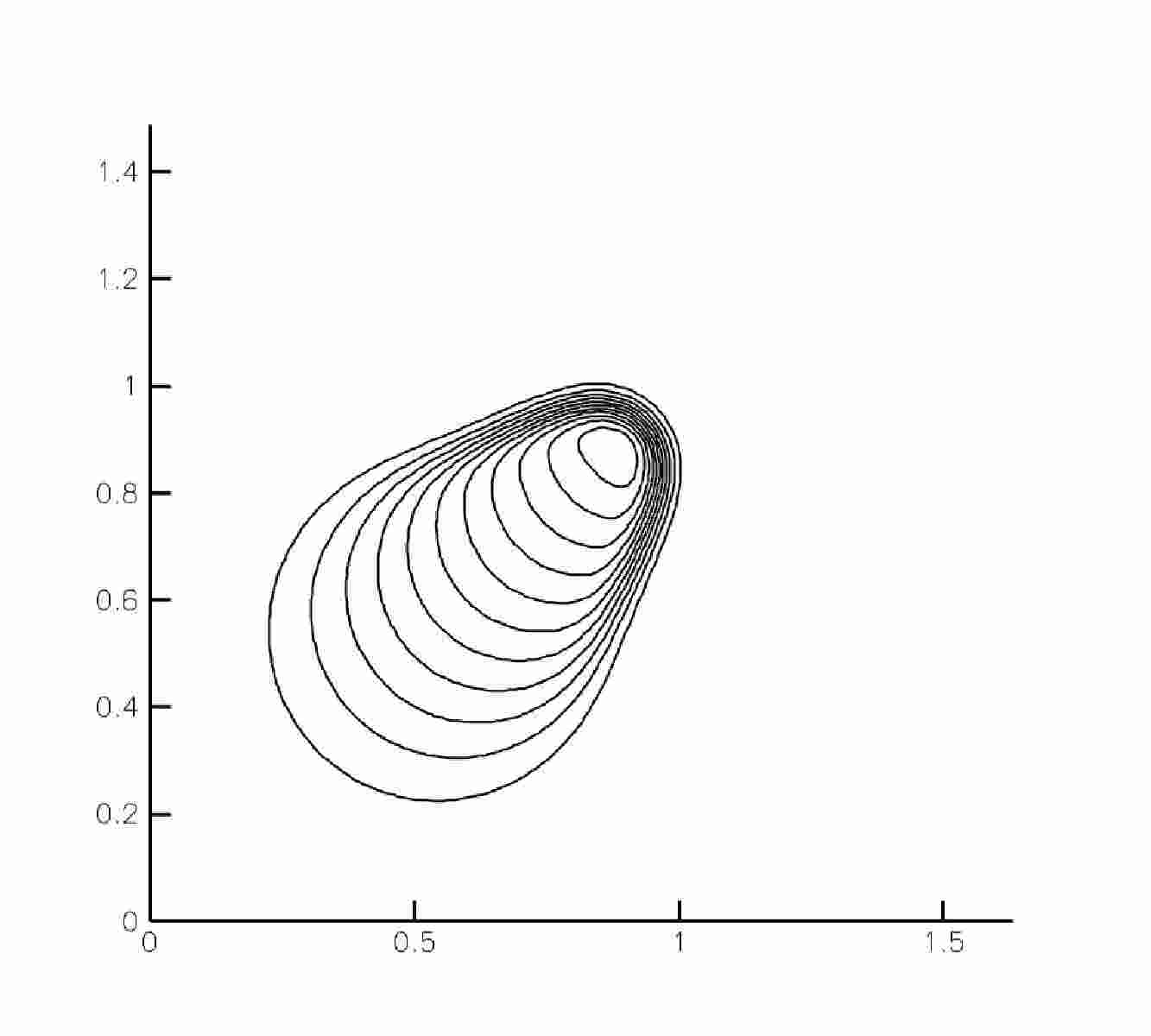}
\end{center} \end{minipage}
\begin{minipage}{0.29\linewidth} \begin{center}
  \includegraphics[width=\linewidth]{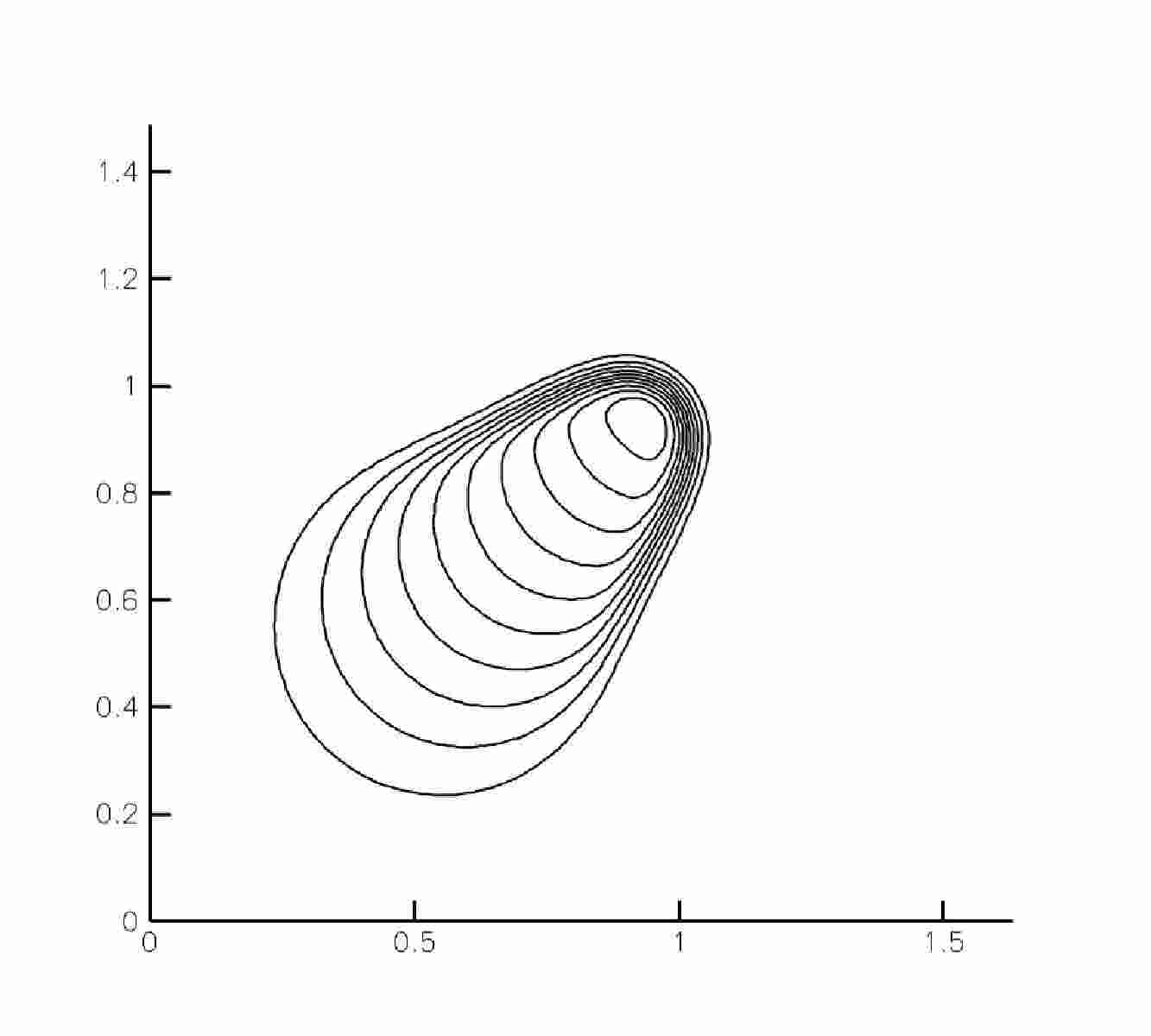}
\end{center} \end{minipage}\\
\begin{minipage}{0.29\linewidth}\begin{center} (d) \end{center}\end{minipage} 
\begin{minipage}{0.29\linewidth}\begin{center} (e) \end{center}\end{minipage} 
\begin{minipage}{0.29\linewidth}\begin{center} (f) \end{center}\end{minipage}
\caption{Contour plots showing the evolution of 2D CFB using a Gaussian filter with a Gaussian pulse as the 
initial condition.  $\alpha=0.08$ with 128 $\times$ 128 grid points.}
\label{2DGaussshock}
\end{center}
\end{figure}

\subsection{Shock Thickness}

One of the characteristics of viscous Burgers equation is that as $\nu$ becomes
smaller, the shocks formed by the solution become thinner and steeper. Heuristically, smaller
$\nu$'s allow more energy into higher wavemodes, which causes a steeper gradient.

Numerical simulations were run for various values of $\alpha$ and various
filters in Equations \eref{CFBa} and \eref{CFBb}, to examine the effects upon shock
thickness. Smaller $\alpha$ correlates with less dampening in the high
wavemodes, thus allowing steeper shocks much like the effect of the viscous
term. Figure \ref{shockthickness}a shows shocks for the Helmholtz filter for
different values of $\alpha$.  As $\alpha$ decreases the shocks get thinner. 
Similar results hold true for all filters and simulations made in two
dimensions.

\begin{figure}[!ht]
\begin{center}
\begin{minipage}{0.48\linewidth} \begin{center}
  \includegraphics[width=.9\linewidth]{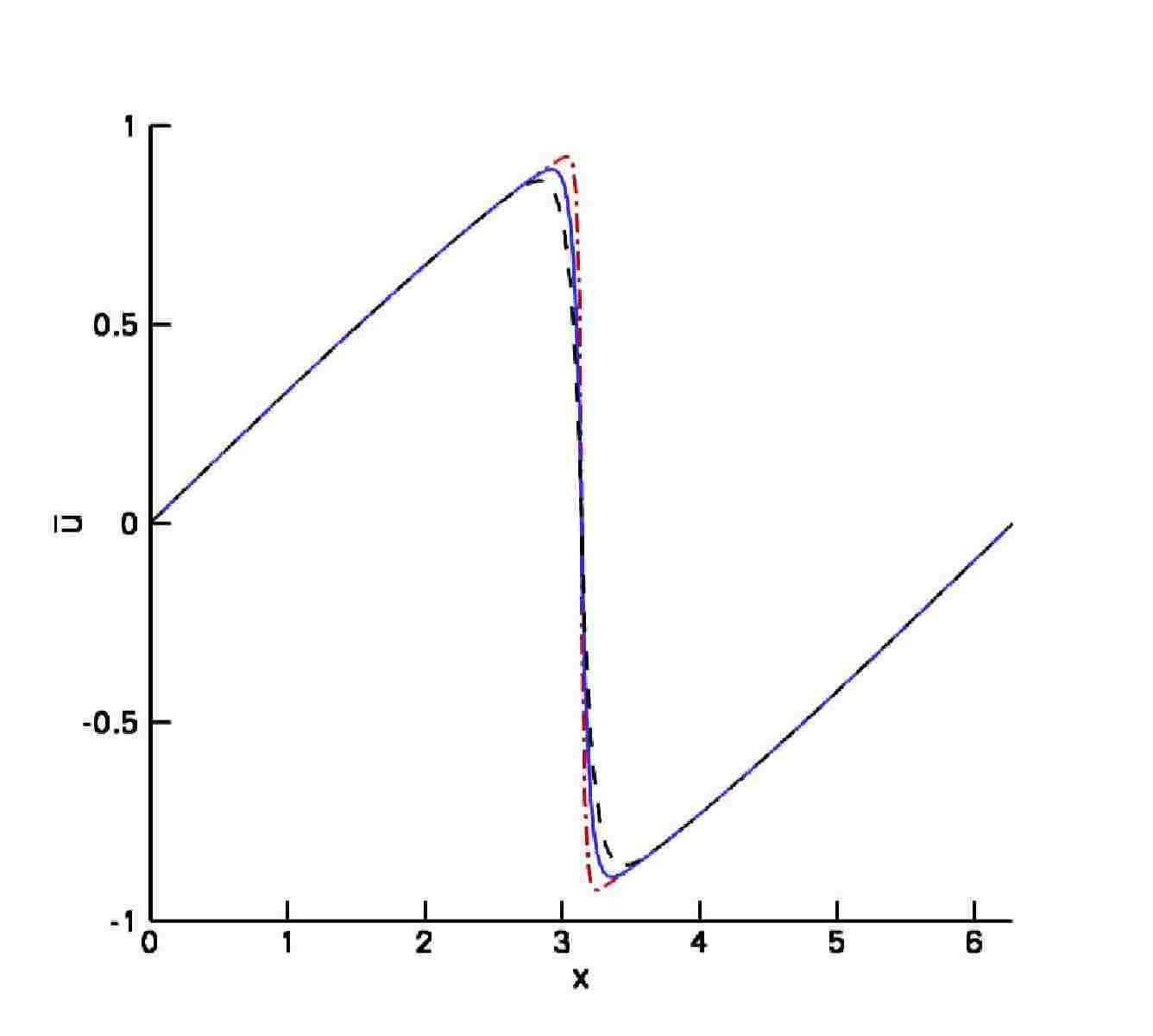}
\end{center} \end{minipage}
\begin{minipage}{0.48\linewidth} \begin{center}
  \includegraphics[width=.9\linewidth]{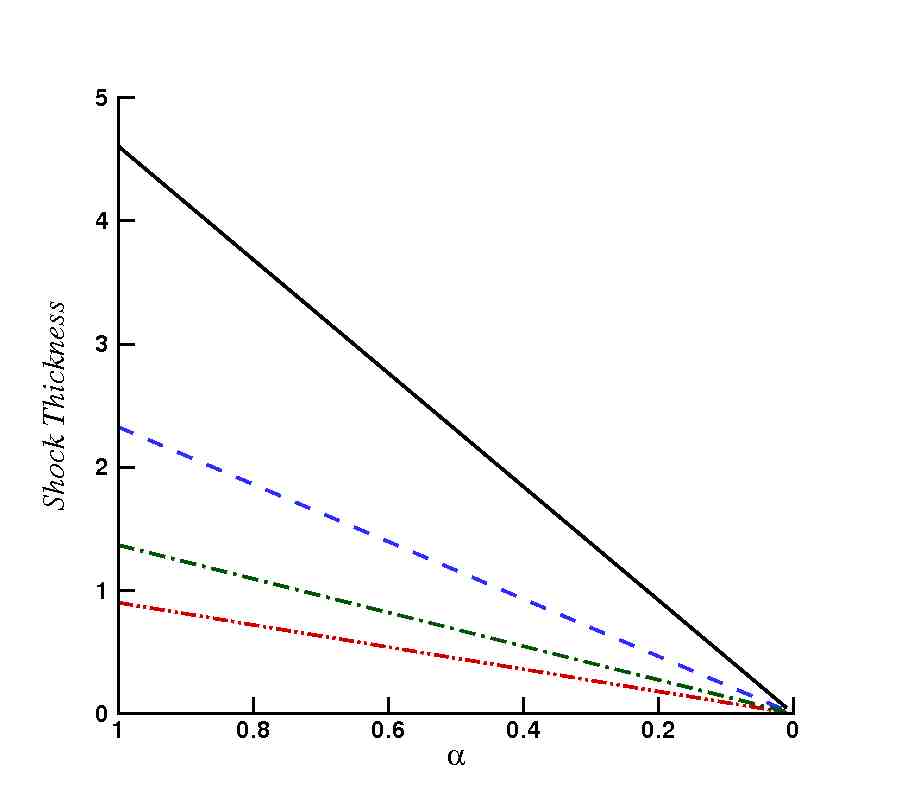}
\end{center} \end{minipage}\\
\begin{minipage}{0.48\linewidth}\begin{center} (a) \end{center} \end{minipage}
\begin{minipage}{0.48\linewidth}\begin{center} (b) \end{center} \end{minipage}
\caption{The thickness of the shocks formed in CFB vary depending upon the value of $\alpha$.  (a) As $\alpha$ decreases the thickness of the shock decreases. $\alpha=0.08$=\dashed, $\alpha=0.05$=\full, $\alpha=0.02$=\dotted (b) The thickness of the traveling shock decreases linearly with $\alpha$. Helmholtz filter,
$g(x)=\frac{1}{2}exp(-|x|)$=\full, Gaussian filter, $g(x)=\pi^{-1/2}exp(-x^2)$=\dashed, Hat filter, $g(x)=\{x-1 \text{ for }
x\in(-1,0),1-x \text{ for } x\in(0,1), 0 \text{ otherwise}\}$=\chain, Tophat filter, $g(x)=\{1 \text{ for } x\in [-\frac{1}{2},\frac{1}{2}], 0 \text{ otherwise} \}$ =\dashddot.}
\label{shockthickness}
\end{center}
\end{figure}

Shock thickness can be examined analytically by looking at the traveling wave solution.   Here the thickness of the shock is defined to be the length over which 90\% of the amplitude change takes place, centered at the center.  From section \ref{travelingwavesolution}, the traveling wave solution is
\begin{equation}
\ubar(x,t)=(u_r-u_l)\int^{x-ct}_{-\infty}g^\alpha(s)\,ds +u_l.
\end{equation}
The thickness of the shock will then be $2\alpha b$, where b is the value where
\begin{equation}
\int_{-b}^b g(x)\, dx=\int_{-\alpha b}^{\alpha b}
\frac{1}{\alpha}g(\frac{x}{\alpha}) \, dx=0.9.
\end{equation}
This length is independent of $u_r$ and $u_l$.  As such, the thickness of the
shock varies linearly on the parameter $\alpha$.  Figure \ref{shockthickness}b
shows shock thickness versus $\alpha$ for different filters.

\section{Spectral Energy}\label{spectralenergy}

Analytically CFB has been examined on an infinite domain.  This is of course
impossible numerically, so in numerical experiments, simulations were performed
on the domains $[0, 2\pi]$ and $[0, 2\pi] \times [0, 2\pi]$ for one and two dimensions, with periodic
boundary conditions.  Since simulations were performed with a pseudospectral method, obtaining the spectral energy
decompositions was easily done.

The examination of the spectral energy decompositions begins by stating a
special case of Sobelev Embedding Theorem that is found in Hunter and
Nachtegaele \cite{HunterJK:01a}.  There it is stated that for a function $f:
\mathbb{T}^n \to \mathbb{C}$ defined by
\begin{equation}
f(\xb)=\sum_{\mb\in \mathbb{Z}^n} a_\mb e^{i \mb \cdot \xb}
\end{equation}
that if 
\begin{equation}
\sum_{\mb\in \mathbb{Z}^n} |\mb|^{2p} |a_\mb|^2 <\infty  
\label{continuitycondition}
\end{equation}
for some $p>\frac{n}{2}$ then $f$ is continuous.  Furthermore if condition
(\ref{continuitycondition}) holds true for $p>j+\frac{n}{2}$ then $f$ has $j$
continuous derivatives.

If spectral energy $E(k)$ is defined as
\begin{equation}
\label{spectralenergydef}
E(k)=\sum_{|\mb| = k} |a_\mb|^2,
\end{equation}
then  \eref{continuitycondition} can then be rewritten as
\begin{equation}
\label{continuitycondition2}
\sum_{k=0}^\infty k^{2p} E(k) < \infty.
\end{equation}

Thus it can be seen that the rate at which $E(k)$ decays as $k \to \infty$, can
determine the smoothness of the equation.  In one dimension, if $E(k)$ decays
faster than $\frac{1}{k^2}$, continuity is guaranteed. In a logarithmic
plot, this correlates with an energy cascade slope of less than -2.  In $n$ dimensions, if $E(k)$ decays faster than $\frac{1}{k^{n+1}}$, (energy cascade slope less than -$(n+1)$) continuity is guaranteed.  Existence of continuous derivatives can be guaranteed in a similar
fashion.

Inviscid Burgers and viscous Burgers (in its inertial frame) have been shown to
have an energy cascade slope of -2 during shock formation \cite{Kraichnan:68a, GurbatovSN:97a}. 
Numerical simulations suggest that CFB also has an energy cascade slope of -2
for the unfiltered velocity, independent of filter and dimension.  The filtered
velocity's energy cascade slope is then highly dependent upon the filter employed. 
In figure \ref{spectralenergyfig}, the spectral energy decompositions for
simulations with the Helmholtz and Gaussian filter in one and two dimensions are
shown. The spectral energy decompositions are shown for a time well after the
shocks are fully developed.  The -2 energy cascade slope is clearly seen, as well
as the effect of the filters on the filtered velocity's energy cascade slope. 
The Helmholtz filtered velocity displays a -6 energy cascade slope for high wave
numbers, guaranteeing a continuous derivative in one and two dimensions. The
Gaussian filtered velocity decays faster than any polynomial for high wave
numbers, thus guaranteeing infinite continuous derivatives. This is similar to the solutions of viscous Burgers equation.  The viscosity term causes an exponential drop in the energy spectrum at high wave numbers, also guaranteeing infinite continuous derivatives.

In section \ref{ExistenceTheorem}, the unfiltered velocity and thus the filtered
velocity were proven to have a continuous derivative for continuously differentiable initial conditions.  In
section \ref{unfilteredvelocity}, it was suggested that despite this, the unfiltered
velocity can form shocks that will grow narrower than any finite resolution. 
Thus for numerical purposes the unfiltered velocity can be considered
discontinuous.  The filtered velocity will appear continuous to numerical
simulations if the energy cascade slope satisfies conditions dictated by the
special case of the Sobolev Embedding Theorem mentioned.

\begin{figure}[!ht]
\begin{center}
\begin{minipage}{0.48\linewidth} \begin{center}
  \includegraphics[width=.9\linewidth]{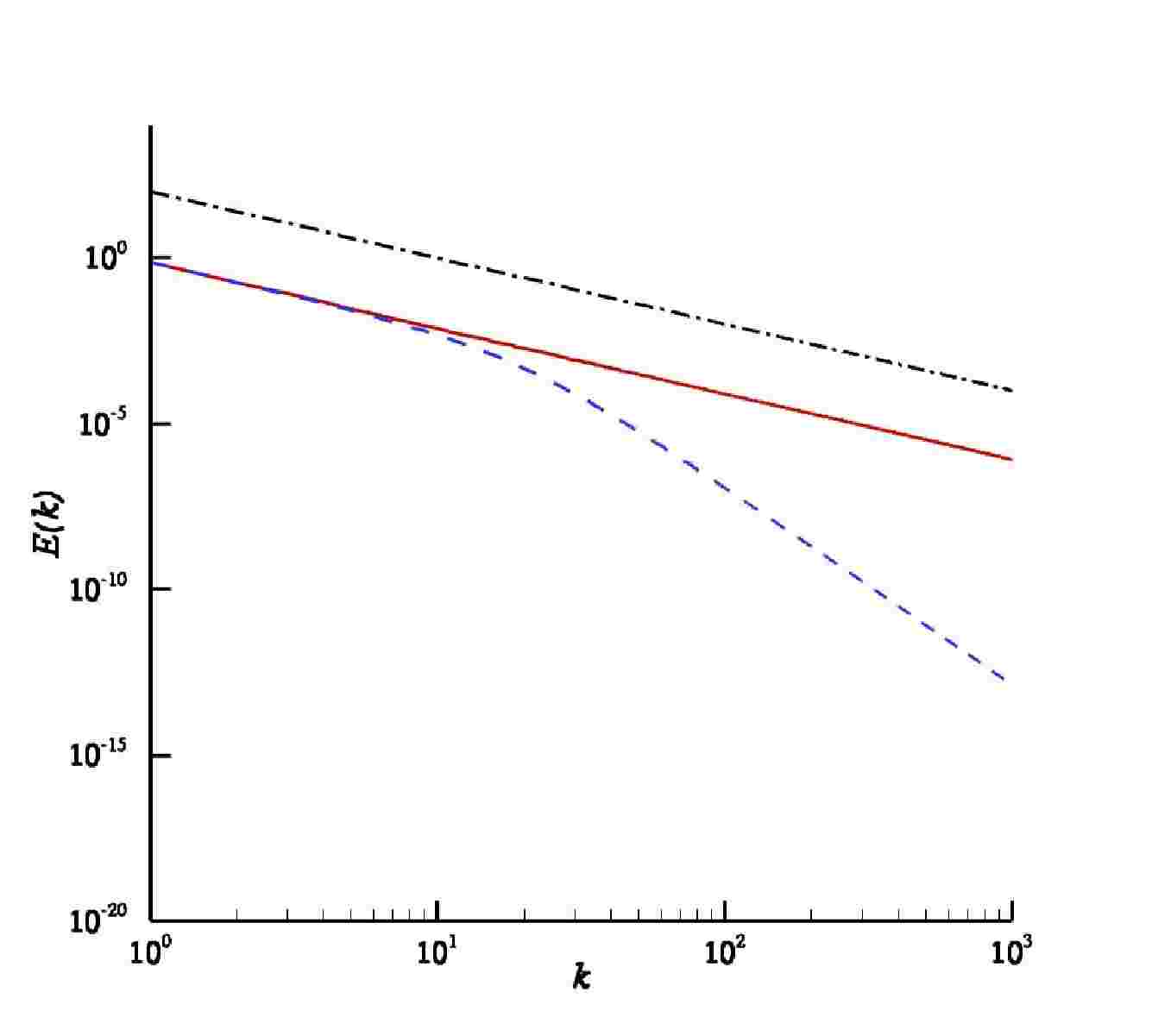}
\end{center} \end{minipage}
\begin{minipage}{0.48\linewidth} \begin{center}
  \includegraphics[width=.9\linewidth]{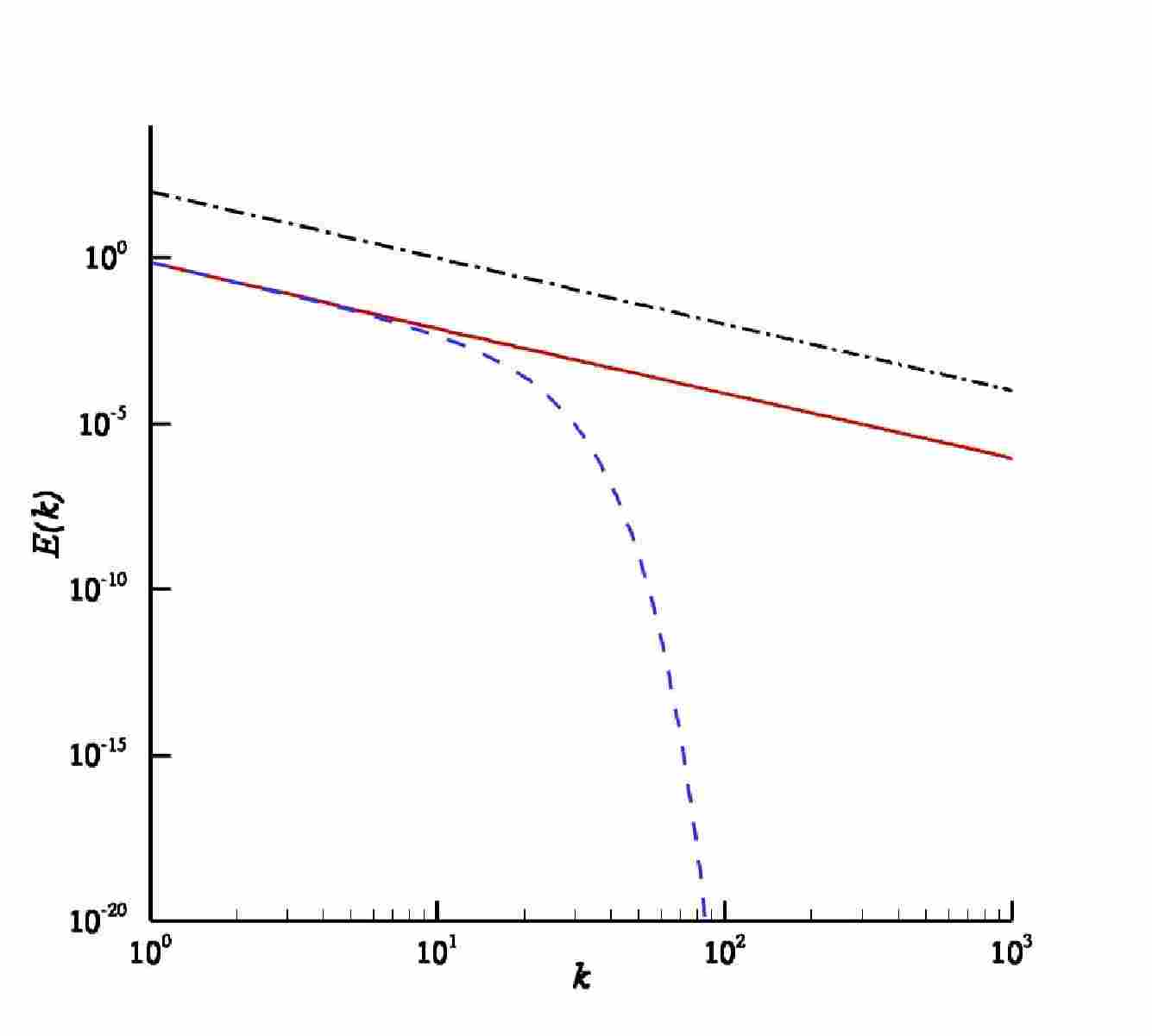}
\end{center} \end{minipage}\\
\begin{minipage}{0.48\linewidth}\begin{center} (a) \end{center} \end{minipage}
\begin{minipage}{0.48\linewidth}\begin{center} (b) \end{center} \end{minipage}
\begin{minipage}{0.48\linewidth} \begin{center}
  \includegraphics[width=.9\linewidth]{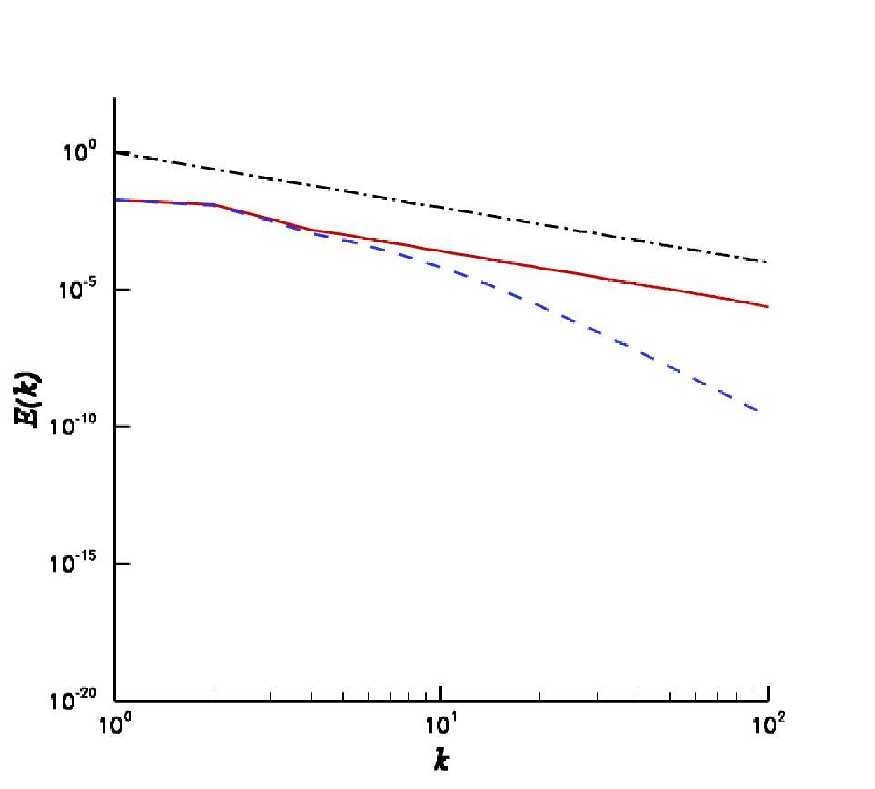}
\end{center} \end{minipage}
\begin{minipage}{0.48\linewidth} \begin{center}
  \includegraphics[width=.9\linewidth]{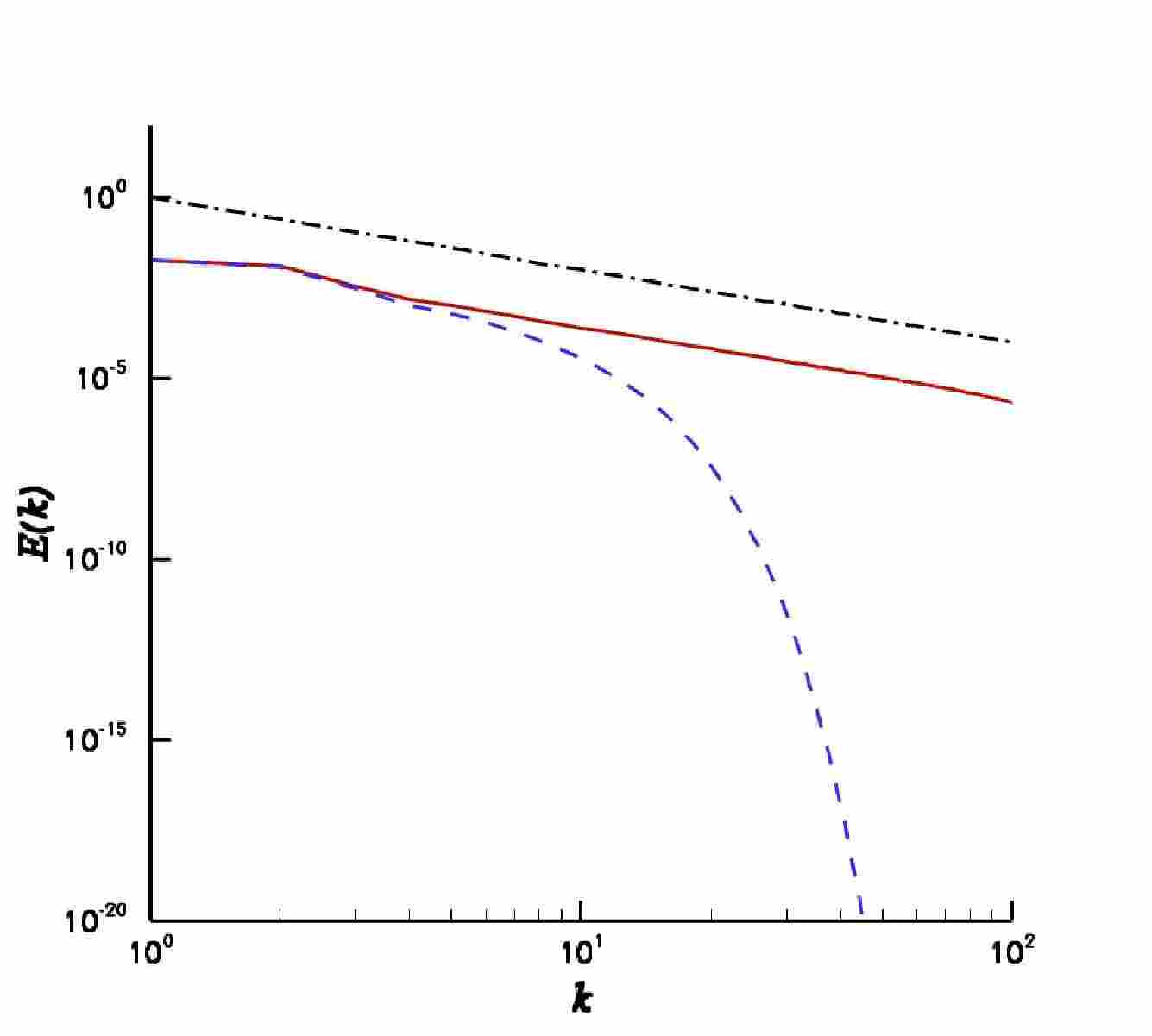}
\end{center} \end{minipage}\\
\begin{minipage}{0.48\linewidth}\begin{center} (c) \end{center} \end{minipage}
\begin{minipage}{0.48\linewidth}\begin{center} (d) \end{center} \end{minipage}
\caption{Spectral energy decompositions for 1D and 2D CFB using Helmholtz and Gaussian filters .  All simulations
were performed with $\alpha = 0.08$. (a) and (b) are one dimensional simulations
at $t=3$ with initial conditions $u(x)=\sin(x)$ for Helmholtz and Gaussian filters respectively.  In figure (a) one can see that the energy cascade slope drops dramatically after wavelength $\frac{1}{\alpha}$.   This occurs similarly in figure (b), with the exception that the spectral energy decreases exponentially due to the Gaussian filter.  (c) and (d) are two
dimensional simulations at $t=1$ with Gaussian pulses as the initial conditions, again with Helmholtz and Gaussian filters respectively.  The spectral energy decompositions show similar characteristics as the one dimensional simulations.  The
spectral energy of $\ub$ \full, $\ubarb$ \dashed, and a reference slope of -2 \chain \, are shown.  }
\label{spectralenergyfig}
\end{center}
\end{figure}

\section{Energy Norms}\label{energynorms}

The kinetic energy for Burgers equation can be defined as $\int \frac{1}{2}\ub
\cdot \ub$.  In CFB, there are two different velocities presenting three
different kinetic energies: $\int \frac{1}{2}\ub \cdot \ub$, $\int
\frac{1}{2}\ub \cdot \ubarb $, and $\int \frac{1}{2}\ubarb \cdot \ubarb$. 
Analytical comparisons between the energy decay rates of viscous Burgers equation and CFB
have proven fruitless with the exception of the one dimensional case using the
Helmholtz filter.  In one dimension, the nonlinear term is
more easily handled, and the Helmholtz filter \eref{helmholtzfilterb} provides a
convenient inverse, allowing greater manipulation capabilities.

For one dimensional viscous Burgers equation, kinetic energy can be defined as
\begin{equation}
E(t)=\int \frac{u (x,t)^2}{2} dx.
\label{vbe}
\end{equation}
The decay rate of the energy is easily calculated to be
\begin{equation}
\label{viscousdecayrate}
\frac{d}{dt}\int \frac{u ^2}{2} dx=-\nu\int u_x^2 dx.
\end{equation}
For one dimensional CFB using the Helmholtz filter, the energy decay rates are 

\label{energydecayfor1Dcrb}
\begin{eqnarray}
\label{uudecayrate}
\frac{d }{d t}\int \frac{u^2}{2} dx=&\alpha^2 \int
\frac{(\ubar_x)^2}{4}(\ubar_x+u_x)dx\\
\label{uubardecayrate}
\frac{d }{d t}\int \frac{\ubar u}{2} dx=&\alpha^2\int(\ubar_x)^3dx\\
\label{ubarubardecayrate}
\frac{d }{d t}\int \frac{\ubar^2}{2} dx=&-3\alpha^2\int \ubar_x
\ubar_{xx} \left(I-\alpha^2 \p x^2 \right)^{-1}(\ubar)dx.
\end{eqnarray}

Here a similarity in structure can be seen, especially when comparing
\eref{viscousdecayrate} with \eref{uubardecayrate}.  The integrands of both are
the first derivatives of velocities, implying that the majority of the energy is lost through steep
gradients, a common concept. The similarity is interesting considering that in
one dimensional viscous Burgers equation, the energy is lost through the viscous term,
while the energy is lost in the nonlinear term for one dimensional CFB.

Numerical simulations were conducted for both viscous Burgers equation and CFB
in one and two dimensions.  Figure \ref{viscousenergy} shows the energy decay
rates for viscous Burgers equation.  Figure \ref{CFBenergy} shows the energy
decay rates for CFB using the Helmholtz and Gaussian filters.  In figure
\ref{CFBenergy}, it can be seen that the different kinetic energies tend toward each other with primary difference occurring in amplitude.  Comparing figures
\ref{viscousenergy} and \ref{CFBenergy}, it is seen that the energy decay rates
behave quantitatively similar between CFB and viscous Burgers equation.  Energy
decay is minimal until the formation of a shock and then the energy decreases
rapidly upon formation of the shock.

\begin{figure}[!ht]
\begin{center}
\begin{minipage}{0.48\linewidth} \begin{center}
  \includegraphics[width=.9\linewidth]{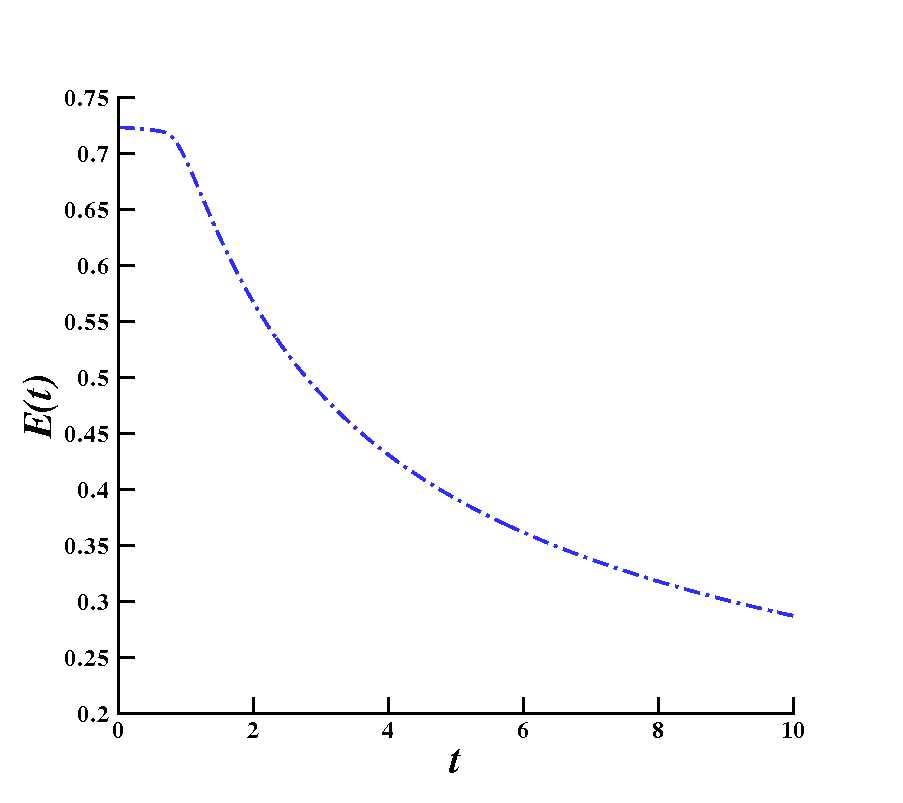}
\end{center} \end{minipage}
\begin{minipage}{0.48\linewidth} \begin{center}
  \includegraphics[width=.9\linewidth]{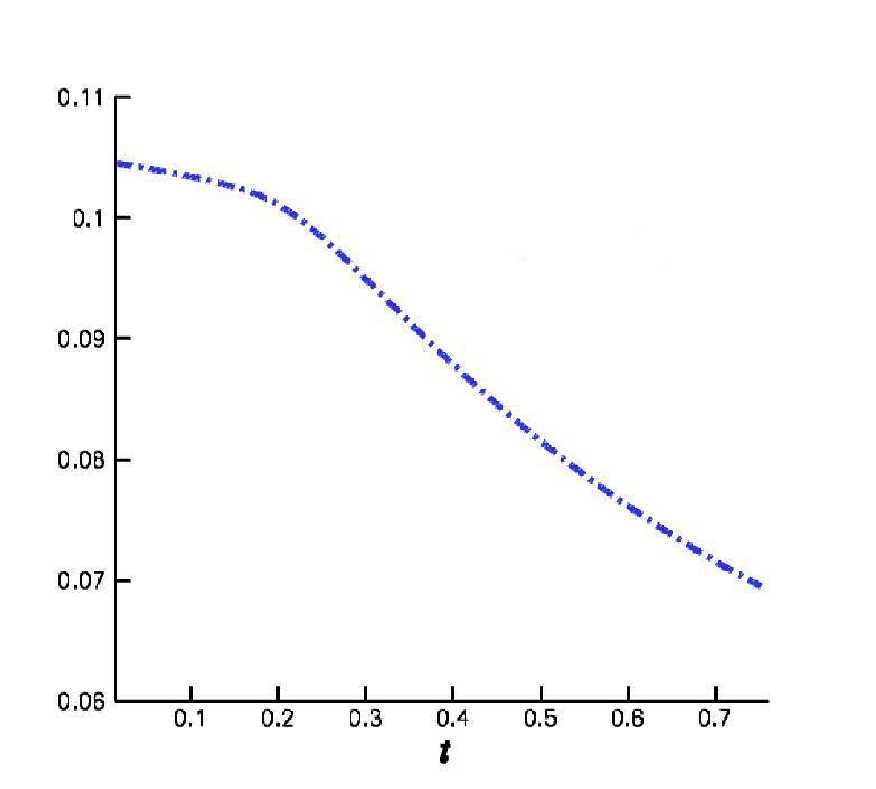}
\end{center} \end{minipage}\\
\begin{minipage}{0.48\linewidth}\begin{center} (a) \end{center} \end{minipage}
\begin{minipage}{0.48\linewidth}\begin{center} (b) \end{center} \end{minipage}
\caption{Energy decay for 1D and 2D viscous Burgers equation. (a) The energy
decay for 1D viscous Burgers equation.  Initial conditions $u(x,0)=exp(-3x^2)$,
with $\nu=0.001$ and 1024 gridpoints.  (b) The energy decay for 2D viscous
Burgers equation.  Initial conditions $u_1(x,y,0)=exp( -30x^2-30y^2)$ and
$u_2(x,y,0)=exp( -30x^2-30y^2)$, with $\nu=0.001$ and 128 $\times$ 128
gridpoints. $\int \ub \cdot \ub$ \chain.}
\label{viscousenergy}
\end{center}
\end{figure}

\begin{figure}[!ht]
\begin{center}
\begin{minipage}{0.48\linewidth} \begin{center}
  \includegraphics[width=.9\linewidth]{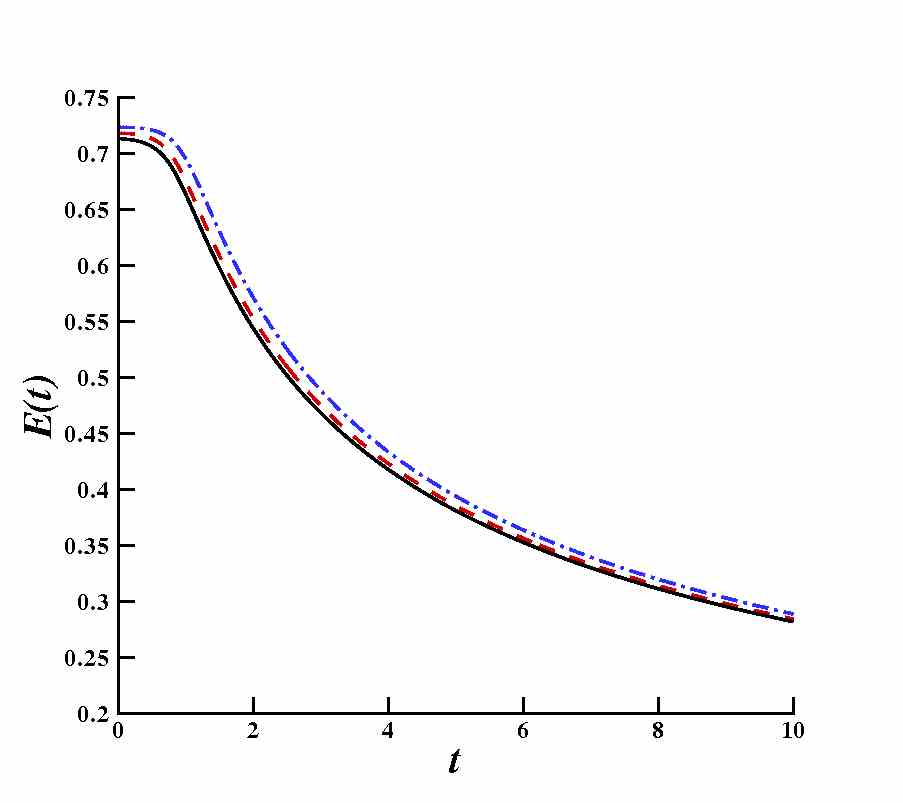}
\end{center} \end{minipage}
\begin{minipage}{0.48\linewidth} \begin{center}
  \includegraphics[width=.9\linewidth]{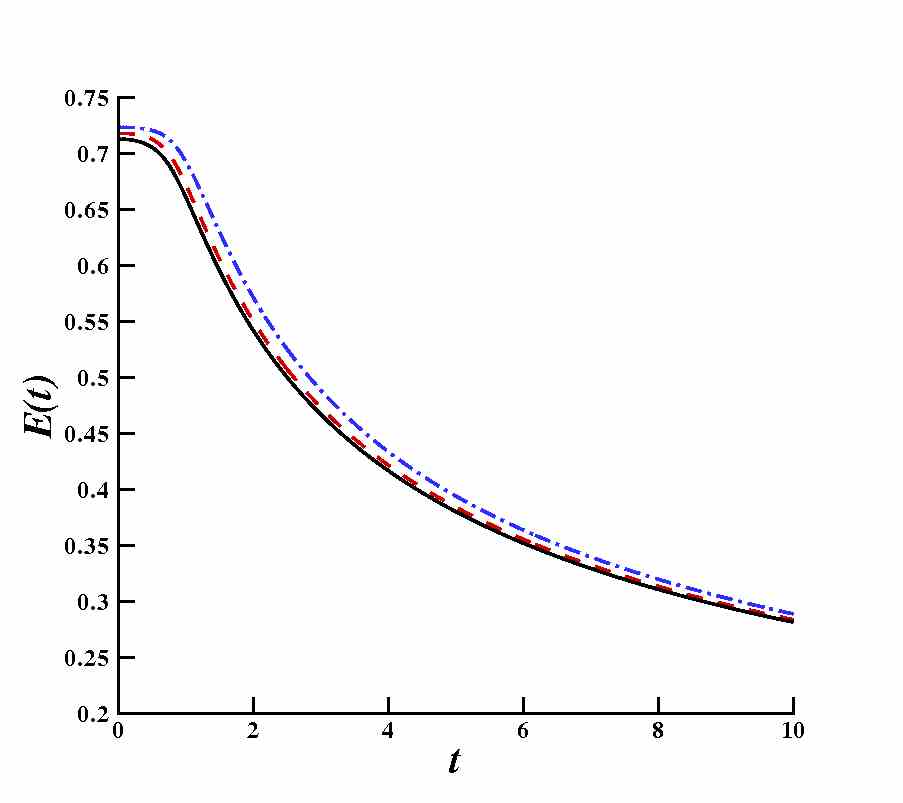}
\end{center} \end{minipage}\\
\begin{minipage}{0.48\linewidth}\begin{center} (a) \end{center} \end{minipage}
\begin{minipage}{0.48\linewidth}\begin{center} (b) \end{center} \end{minipage}
\begin{minipage}{0.48\linewidth} \begin{center}
  \includegraphics[width=.9\linewidth]{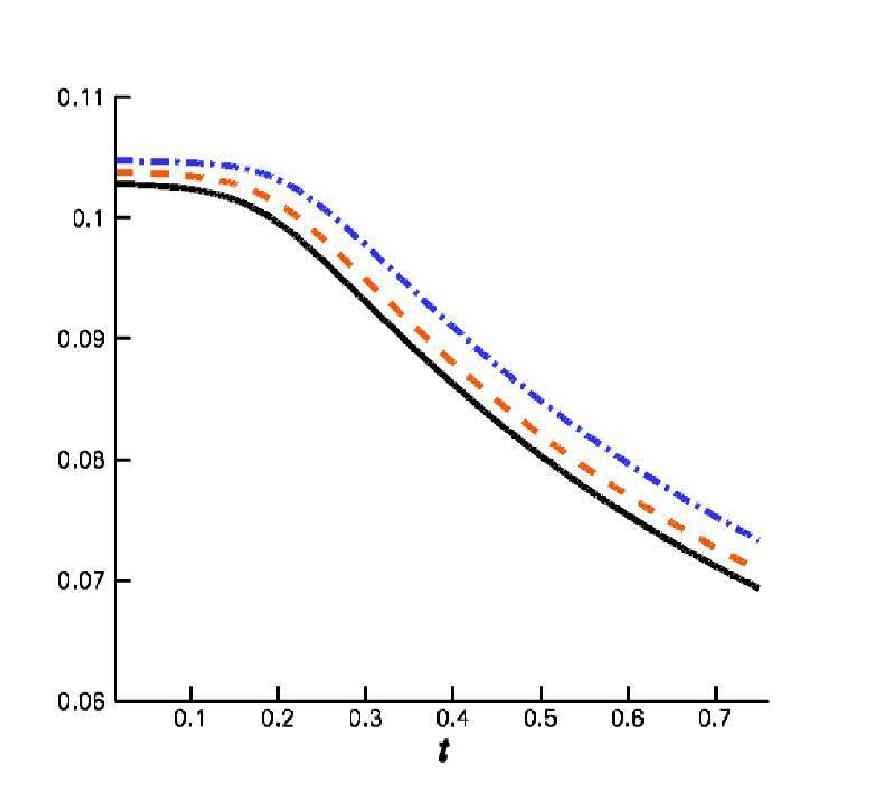}
\end{center} \end{minipage}
\begin{minipage}{0.48\linewidth} \begin{center}
  \includegraphics[width=.9\linewidth]{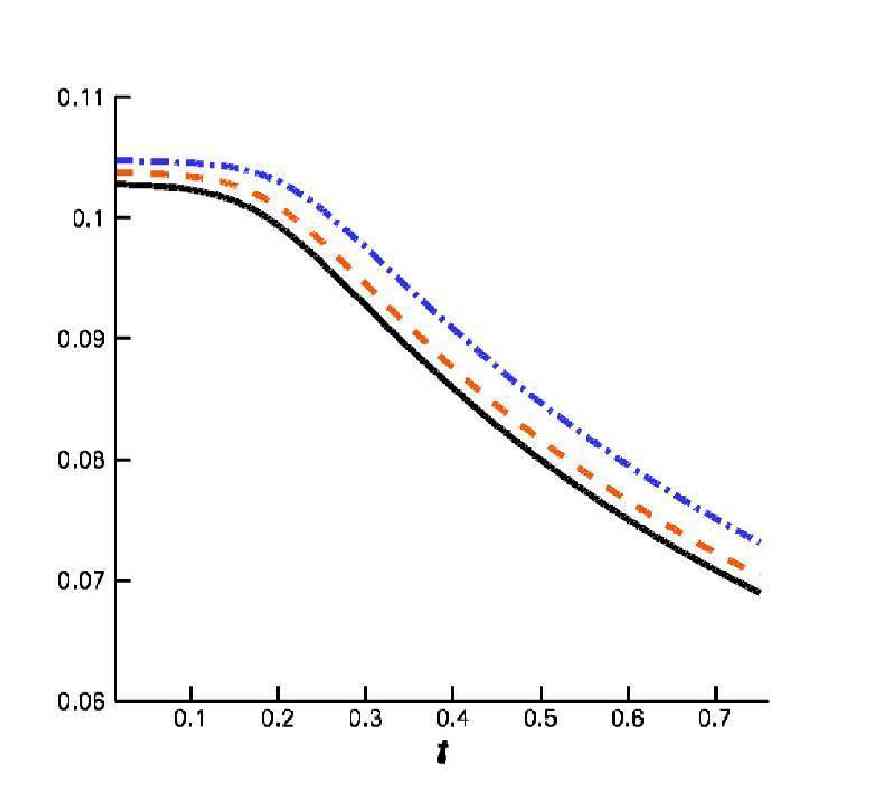}
\end{center} \end{minipage}\\
\begin{minipage}{0.48\linewidth}\begin{center} (c) \end{center} \end{minipage}
\begin{minipage}{0.48\linewidth}\begin{center} (d) \end{center} \end{minipage}
\caption{Energy decay for 1D and 2D CFB. $\int \ub \cdot \ub$ \chain, $\int \ub \cdot \ubarb$
\dashed,  $\int \ubarb \cdot \ubarb$ \full.  Figures (a) and (b) show the energy decay for 1D CFB 
with the Helmholtz and Gaussian filter respectively. Initial conditions $u(x,0)=exp(-3x^2)$, with $\alpha=0.05$ and 1024 gridpoints. Figures (c) and (d) show the energy decay for 2D CFB with the Helmholtz and Gaussian filter.  Initial conditions $u_1(x,y,0)=exp(-30x^2-30y^2)$ and $u_2(x,y,0)=exp( -30x^2-30y^2)$, with $\alpha=0.05$ and 128
$\times$ 128 gridpoints. Decay rates are similar to those seen in figure \ref{viscousenergy}.  }
\label{CFBenergy}
\end{center}
\end{figure}

\subsection{Energy Can Increase}

It is also important to notice that Equations \eref{uudecayrate}, \eref{uubardecayrate}, and \eref{ubarubardecayrate} are not sign definite as is Equation \eref{vbe}.  This shows while the energy for viscous Burgers equation must always decrease, CFB can experience an increase in energy. Specifically, it can be seen in Equation \eref{uubardecayrate}, that for steep decreasing gradients energy will be lost, but for steep increasing gradients energy will be gained.  However, due to the nature of the equation increasing gradients will decrease in steepness, while decreasing gradients will increase or remain steep.  This can be seen in the evolution of a Gaussian pulse, as in figure \ref{filtershock}. Thus any increase in energy should be brief.  

One would also suspect that as $\alpha \to 0$ that the brief energy increase would dissappear. If solutions to CFB are, in fact, limiting to solutions of inviscid Burgers equation, then it follows that the energy rates would approach those of inviscid Burgers equation, as well.

Figure \ref{positiveenergy} shows two numerical simulations where the initial conditions were chosen to generate an increase in energy.  In both cases the energy will briefly rise before beginning to decay as expected.  Figure \ref{positiveenergy2} shows three simulations with decreasing $\alpha$.  As was expected, the brief increase in energy becomes less dramatic as $\alpha$ decreases.

\begin{figure}[!ht]
\begin{center}
\begin{minipage}{0.48\linewidth} \begin{center}
  \includegraphics[width=.9\linewidth]{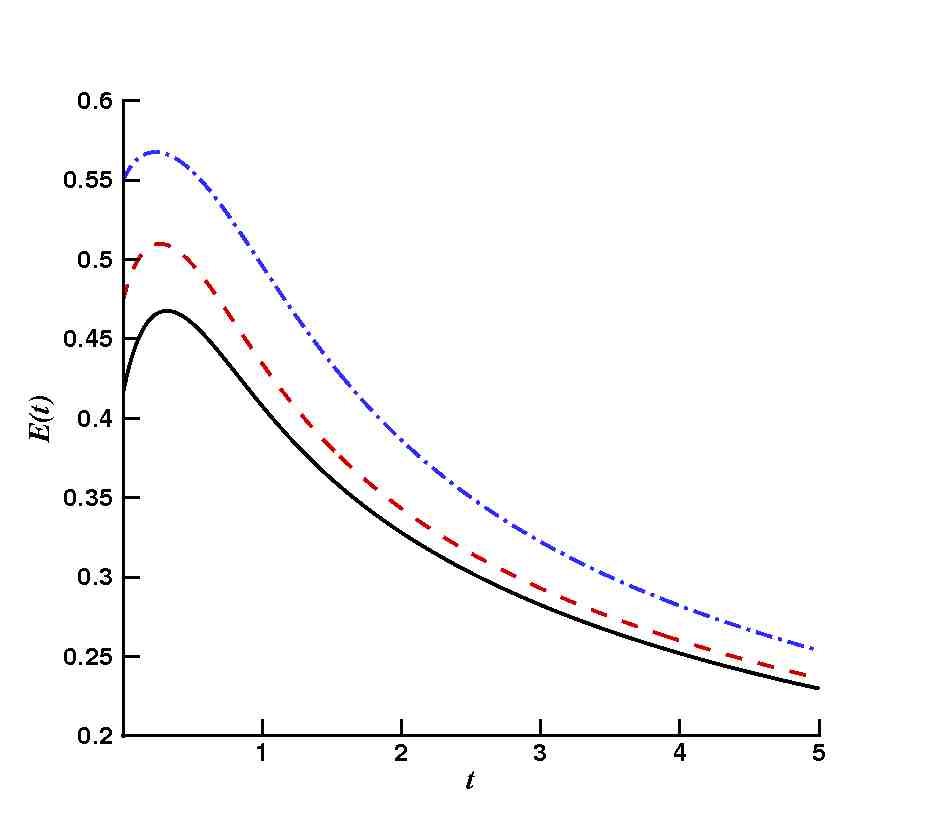}
\end{center} \end{minipage}
\begin{minipage}{0.48\linewidth} \begin{center}
  \includegraphics[width=.9\linewidth]{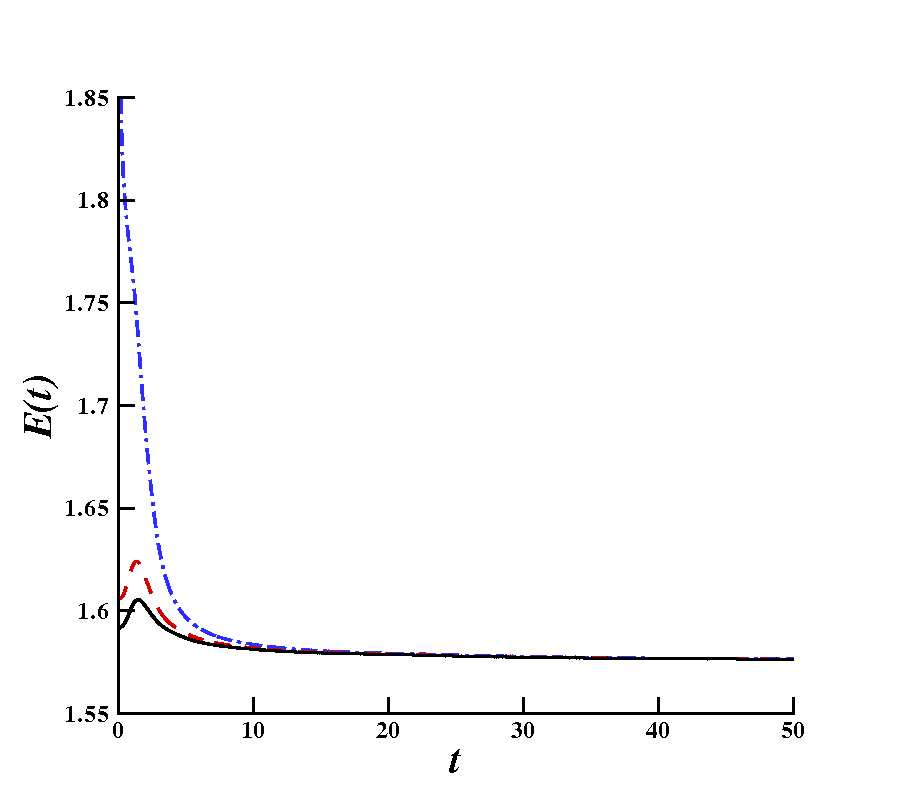}
\end{center} \end{minipage}\\
\begin{minipage}{0.48\linewidth}\begin{center} (a) \end{center} \end{minipage}
\begin{minipage}{0.48\linewidth}\begin{center} (b) \end{center} \end{minipage}
\caption{Here the energies for CFB are shown. $\int
\ub \cdot \ub$ \chain, $\int \ub \cdot \ubarb$ \dashed,  $\int \ubarb \cdot \ubarb$
\full.   In (a) the initial conditions are $u_0=C(x -\pi)/(1+100(x-\pi)^4) $, with
$C$ chosen such that $\max(u_0) =1$, and (b) with a random initial condition. 
In both cases the energies increase initially, but then behave with normal
energy decaying behavior.}
\label{positiveenergy}
\end{center}
\end{figure}

\begin{figure}[!ht]
\begin{center}
\includegraphics[width=0.9\linewidth]{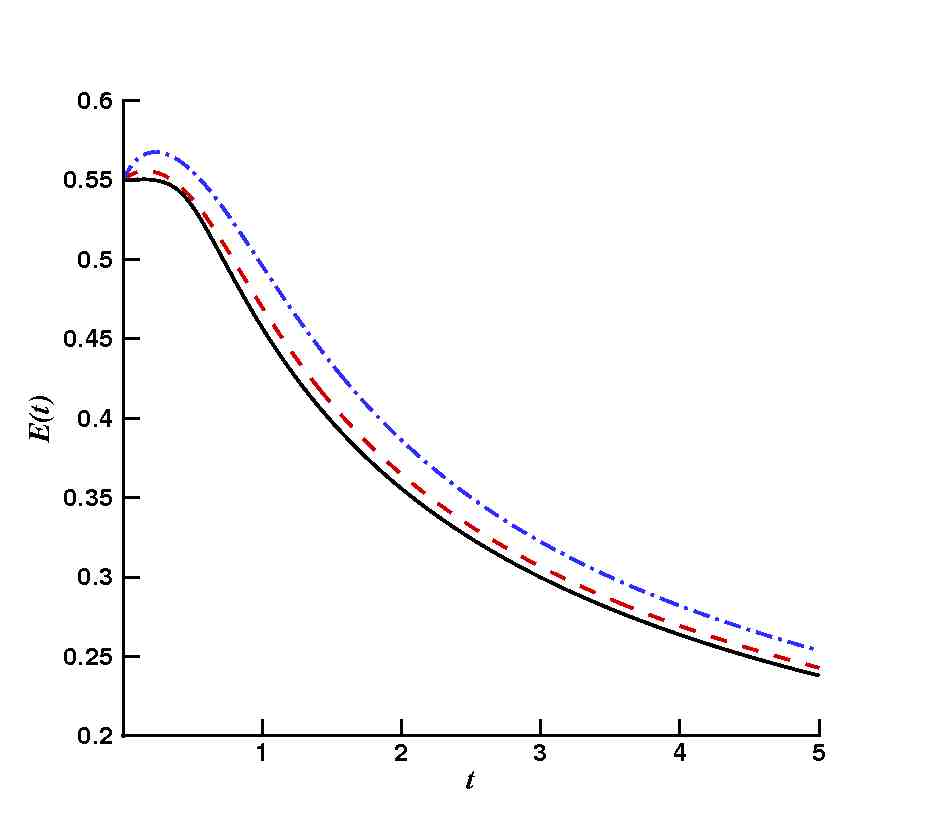}
\caption{ The energy $\int \ub \cdot \ub$ is compared for three different values of $\alpha$. $\alpha=0.08$ \chain, $\alpha=0.05$ \dashed, and $\alpha=0.02$ \full.  The initial conditions are $u_0=C(x -\pi)/(1+100(x-\pi)^4) $, with $C$ chosen such that $\max(u_0) =1$. For all three values of $\alpha$ there is a brief increase of energy, but as $\alpha$ becomes smaller, this increase in energy becomes less substantial.  }
\label{positiveenergy2}
\end{center}
\end{figure}

\section{Conclusion}
By passing the convective velocity through a low pass filter, the cascade of energy into higher wave modes of inviscid Burgers equation has been altered.  For a general class of filters, continuous initial conditions lead to continuous solutions for all time.  Much of the geometric structure seen in invariants of motion is preserved in the averaging process.  We can see that traveling wave solutions propagate at the correct wave speed, and converge to weak solutions of inviscid Burgers equation.  Through numerical simulations, much of the behavior of CFB can be examined.  The shock formation and behavior  of the unfiltered velocity appear similar to that of viscous Burgers equation.  The shock thickness can be regulated by the parameter $\alpha$, which is the characteristic width of the filter.  Spectral energy decompositions give insight into how the highwave modes are handled and what filters are needed to guarantee the smoothness of the solution.  Energy norms also compare favorably to those found in viscous Burgers equation.  All together, filtering the convective velocity appears to be a valid technique for high wavemode regularization and can hopefully be extended beyond Burgers equation into more general fluid dynamics equations.

\section{Acknowledgments}
The research in this paper was partially supported by the AFOSR contract FA9550-05-1-0334.

\bibliography{RefA1}
\bibliographystyle{unsrt}

\end{document}